\documentclass[pdftex,11pt]{article}%
\usepackage[pdftex]{hyperref}%
\usepackage[pdftex]{graphicx}%
\usepackage[normalem]{ulem}
\usepackage{amsmath,amsthm,amsfonts,
amsbsy,amssymb,upref,enumerate,bigstrut,
color,mathtools,mathrsfs,float,bm,dsfont}
\usepackage{lipsum}
\usepackage{stmaryrd, url,caption}
\usepackage{authblk,subfigure} 

\usepackage[left=1.0in,top=1.3in,right=0.8in]{geometry}

\usepackage[figurename=Fig.]{caption}

\newtheorem{theorem}{Theorem}[section]

\newtheorem{corollary}[theorem]{Corollary}

\newtheorem{definition}{Definition}[section]
\newtheorem{lemma}[theorem]{Lemma}
\newtheorem{proposition}[theorem]{Proposition}

\numberwithin{equation}{section} 

\makeatletter
\def\@seccntformat#1{\@ifundefined{#1@cntformat}%
	{\csname the#1\endcsname\quad}
	{\csname #1@cntformat\endcsname}
}
\makeatother

\newif\ifShowComments
\ShowCommentstrue
\def\strutdepth{\dp\strutbox}
\def\druk#1{\strut\vadjust{\kern-\strutdepth
        {\vtop to \strutdepth{%
                \baselineskip\strutdepth\vss
                        \llap{\hbox{#1}\quad}\null}}}}

\title{\bf
A Bimodal Model for Extremes Data
}

\author[1]{Cira E. G. Otiniano \thanks{ciragotiniano@gmail.com}}
\author[1]{Bianca Sousa  \thanks{bianca.unb@gmail.com} }
\author[1]{Roberto Vila \thanks{rovig161@gmail.com}}
\author[2]{Marcelo Bourguignon \thanks{m.p.bourguignon@gmail.com}}
\affil[1]{Departamento de Estat\'istica, Universidade de Bras\'ilia, 70910-900, Bras\'ilia, Brazil}
\affil[2]{Departamento de Estat\'istica, Universidade Federal do Rio Grande do Norte, 59078-970, Natal/RN, Brazil}

\begin{document}
\maketitle

\begin{abstract}
{

In extreme values theory, for a sufficiently large block size, the maxima distribution is approximated by the generalized
extreme value (GEV) distribution. The GEV distribution is a family of continuous probability distributions, which has wide applicability in several areas including hydrology, engineering, science, ecology and finance. However, the GEV distribution is not suitable to model extreme bimodal data.
In this paper, we propose an extension of the GEV distribution that incorporate an additional
parameter. The additional parameter introduces bimodality and to vary tail weight, i.e.,
this proposed extension is more flexible than the GEV distribution.
Inference for the proposed distribution were performed under the likelihood
paradigm. A Monte Carlo experiment is conducted to evaluate the performances of these
estimators in finite samples with a discussion of the results. Finally,
the proposed distribution is applied to environmental data sets, illustrating
their capabilities in challenging cases in extreme value theory.

}
\end{abstract}


%
\section{Introduction}
\noindent

The  generalized extreme value (GEV) distribution is widely used to model extreme events in several areas such as finance, insurance, hydrology, bioinformatics, among others. In journals of statistics and applied areas, a huge amount of works with applications of the GEV distribution can be found. Theory and applications of extreme value distributions can be found in the books by Kotz and Nadarajah (2000)
, Coles (2001)
, Beirlant et al. (2004)
,  Haan and Ferreira (2010) 
, Embrechts et al. (2013), Longin (2016)
, Scheirer (2017) 
, among others.

A random variable $Y$ has a Generalized Extreme Value (GEV) distribution with shape parameter $\xi$, location parameter $\mu\in\mathbb{R}$, and scale paramer $\sigma>0$  denoted by
$X\sim F(\cdot;\xi,\mu,\sigma)$, if its probability density function (PDF) is given by
\begin{eqnarray*}
	f_{}(y;\xi,\mu,\sigma)
	=
	\begin{cases}
		\displaystyle
		\frac{1}{\sigma}\Big[ 1+\xi \Big(\frac{y-\mu}{\sigma} \Big)\Big]^{(-1/\xi)-1}\,
		\exp\biggl\{ -\Big[ 1+\xi \Big(\frac{y-\mu}{\sigma} \Big)\Big]^{-1/\xi} \biggr\}, & \xi\neq 0,
		\\[0,4cm] \noindent
		\displaystyle
		\frac{1}{\sigma}\exp\biggl\{-\Big(\frac{y-\mu}{\sigma} \Big)-\exp\Big[-\Big(\frac{y-\mu}{\sigma} \Big)\Big]\biggr\},
		& \xi=0,
	\end{cases}
	\quad
\end{eqnarray*}
valid for $y>\mu-\sigma/\xi$ in the case $\xi>0$, and $y<\mu-\sigma/\xi$ in the case $\xi<0$.

The cumulated distribution function (CDF) of the GEV distribution is given by
\begin{eqnarray}\label{cdfgev1}
	F_{}(y;\xi,\mu,\sigma)
	=
	\begin{cases}
		\displaystyle
		\exp\biggl\{ -\Big[ 1+\xi \Big(\frac{y-\mu}{\sigma} \Big)\Big]^{-1/\xi} \biggr\}, & \xi\neq 0,
		\\[0,4cm]
		\displaystyle
		\exp\biggl\{-\exp\Big[-\Big(\frac{y-\mu}{\sigma} \Big)
		\Big]  \biggr\}, & \xi=0.
	\end{cases}
\end{eqnarray}
When $\xi=0$ the pdf \eqref{cdfgev1} is also known as the Gumbel distribution.

Here  a GEV random variable with $\xi \neq 0$  and standar scale we denote by $Y \sim F_{\rm G}(\cdot; \xi, \mu) =F_{}(\cdot;\xi,\mu,1)$. In this case, the PDF and CDF functios are given, respectively, by
\begin{eqnarray}\label{pdfgev2}
	f_{\rm G}(y;\xi,\mu)=	\left[ 1+\xi (y-\mu)\right]^{(-1/\xi)-1}\,
	\exp\left\{ -\left[ 1+\xi (y-\mu)\right]^{-1/\xi} \right\}
\end{eqnarray}
and
\begin{eqnarray*}
	F_{\rm G}(y;\xi,\mu)=	\exp\left\{ -\left[ 1+\xi (y-\mu)\right]^{-1/\xi} \right\}.
\end{eqnarray*}

In extreme value modeling, in special, environmental data (see Section 7)
data with bimodality has appeared more and more. The statistical modeling of this type of data requires distributions that capture bimodality.
In this context, Nascimento et al. (2016) 
provided some new extended models to the GEV distribution as a baseline function and the transmuted GEV distribution  introduced by Aryal and Tsokos (2009) 
, showing advantages compared with the standard GEV distribution. 
Eljabri and Nadarajah (2017) studied the Kumaraswamy GEV distribution.
However, all the models cited above are
not suitable for capturing bimodality.
In this sense, the mixture of GEV distributions and the mixture of extreme value distributions has been an alternative ( Otiniano et al. (2016) 
). The difficulty with these models is in the estimation process, as six sub-models must be considered.

This paper presents a new model for extremes, based on transformation of the
standard GEV, with the hope it yields a ``better fit'' in certain extremes analysis
(see Sect. \ref{sec:5}). The inclusion of additional parameter in the GEV role is to
introduce bimodality, that is, the new distribution is more flexible
than the GEV distribution, and without computational complications for the estimation of its parameters in the case $\xi\neq 0$. 

In Sect. \ref{sec:2}, we define the BGEV model. In Sect. \ref{sec:3},
we derive the main properties of the bimodal GEV distribution.
In Sect. \ref{sec:4}, we provide graphical illustrations.
In Sect. \ref{sec:5}, inference procedure is carried out under the
likelihood paradigm.
Already, in Sect. \ref{sec:6}
discusses some simulation results for the estimation method.
In Sect. \ref{sec:7}, two illustrative applications
in environmental data sets are investigated.
Conclusions are addressed in Sect. \ref{sec:8}.

\section{The Bimodal GEV distribution}
\label{sec:2}
\noindent

In this section,  we introduce a bimodal distribution of extreme value (BGEV)
as Rathei and Swamme (2007) 
.

We define a random variable with BGEV distribution and parameters $\xi\neq 0, \mu, \sigma, \delta$,
denoted $X\sim F_{\rm BGEV}(\cdot; \xi,\mu,\sigma, \delta)$,
if its PDF is given by
\begin{align}\label{pdfbgev1}
	f_{\rm BGEV}(x;\xi,\mu,\sigma, \delta)&=f_{\rm G}\big(T_{ \sigma,\delta}(x);\xi, \mu\big)\, T_{ \sigma,\delta}'(x)
	\nonumber
	\\[0,1cm]
	&=
	\sigma(\delta+1)|x|^\delta \big[ 1+\xi (\sigma x|x|^\delta-\mu)\big]^{(-1/\xi)-1}
	\exp\left\{ -\big[ 1+\xi (\sigma x|x|^\delta-\mu)\big]^{-1/\xi} \right\},
\end{align}
where
\begin{align}\label{trans}
T_{\sigma,\delta}(x)
=
\sigma x|x|^\delta \,, \quad x\in\mathbb{R},  \ \  \delta\geqslant -1,   \ \ \sigma> 0
\end{align}
is an invertible transformation with derivatives
\begin{align}\label{t-lin}
T_{ \sigma,\delta}'(x)= \sigma(\delta+1)|x|^\delta\,,
\quad
T_{\sigma,\delta}''(x)= {\rm sign}(x)\, \sigma (\delta+1)\delta|x|^{\delta-1}\,,
\end{align}
and
\[
T_{ \sigma,\delta}^{(k)}(x)
=
\big[{\rm sign}(x)\big]^{k-1}\, \sigma \, \Bigg[\prod_{i=-1}^{k-2}(\delta-i)\Bigg] |x|^{\delta-(k-1)}\,,
\quad
k\geqslant 2\,,
\]
where ${\rm sign}(\cdot)$ denotes the sign function.

\smallskip

Since  $T_{ \sigma,\delta}$ is non-decreasing for $ \delta> -1 $,  then, its inverse function, denoted by $T_{ \sigma,\delta}^{-1}$, exists and is given by
\begin{align}\label{tinv}
T_{\sigma,\delta}^{-1}(x)={\rm sign}(x)
\biggl( \frac{|x|}{\sigma}\biggr)^{1/(\delta+1)}.
\end{align}

The PDF (\ref{pdfbgev1}) is valid for $x>T_{ \sigma,\delta}^{-1}(\mu-1/\xi)$ in the case $\xi>0$, and for
$x<T_{ \sigma,\delta}^{-1}(\mu-1/\xi)$ in the case $\xi<0$.

\smallskip
Note that, for $ \xi> 0 $, the function  $ f _{\rm BGEV} $ is  a PDF, because
\begin{align*}
\int_{T_{\sigma,\delta}^{-1}(\mu-1/\xi)}^{\infty} f_{\rm BGEV}(x;\xi,\mu,\sigma, \delta) \, {\rm d}x
=
\int_{\mu-1/\xi}^{\infty} f_{\rm G}\big(T_{ \sigma,\delta}(x);\xi, \mu \big)\, {\rm d}T_{\sigma,\delta}(x)
=
1\,.
\end{align*}
Similarly, it is verified that, for $\xi<0$, $f_{\rm BGEV}$ is a PDF.

\smallskip

The expression \eqref{pdfbgev1} is equivalent to
\[
f_{\rm BGEV}(x;\xi,\mu,\sigma, \delta)
=
(F_{\rm G}\circ T_{ \sigma,\delta})'(x )\,,
\]
where
\begin{equation}\label{cdfbgev1}
(F_{\rm G}\circ T_{ \sigma,\delta})(x; \xi, \mu)=F_{\rm BGEV}(x; \xi,\mu,\sigma, \delta)
\end{equation}
is the CDF of $X$.

When  $\delta=0$  the  function \eqref{pdfbgev1} is of type  \eqref{cdfgev1};  $F_{\rm BGEV}(x; \xi,\mu,\sigma, 0)= F_{\rm G}(x; \xi,\mu/\sigma,1/\sigma)$.


\section{Some properties of the BGEV distribution}
\label{sec:3}
\noindent

In this section, some mathematical properties as monotonicity,  bimodality property, stochastic representation, moments, quantiles and  tail behavior of the BGEV distribution are discussed.
\subsection{Monotonicity}
\noindent

Let $m_\xi$ be the unique mode for the GEV distribution \eqref{pdfgev2}.
To state the next result we define the following quantities:
$x^{\min}_\xi=\min\{0,T_{\sigma,\delta}^{-1}(m_\xi)\}$ and
$x_\xi^{\max}=\max\{0,T_{ \sigma,\delta}^{-1}(m_\xi)\}$.
\begin{proposition}
The following monotonicity properties it hold:
\begin{itemize}
\item[\rm 1)] If $\xi> 0$ and $\delta<0$, then the {\rm BGEV} PDF  is increasing for each $x<x^{\min}_\xi$;
\item[\rm 2)] If $\xi> 0$ and $\delta>0$, then the {\rm BGEV} PDF  is decreasing for each $x<x^{\min}_\xi$;
\item[\rm 3)] If $\xi< 0$ and $\delta<0$, then the {\rm BGEV} PDF  is decreasing for each $x>x^{\max}_\xi$;
\item[\rm 4)] If $\xi< 0$ and $\delta>0$, then the {\rm BGEV} PDF  is increasing for each $x>x^{\max}_\xi$.
\end{itemize}
\end{proposition}
\begin{proof}
Since $T_{ \sigma,\delta}'(x)= \sigma(\delta+1)|x|^\delta$, for $\delta<0$, $T_{\sigma,\delta}'(x)$ is decreasing (resp. increasing) for each $x>0$ (resp. $x<0$). For $\delta>0$, $T_{\sigma,\delta}'(x)$ is increasing (resp. decreasing) for each $x>0$ (resp. $x<0$).
On the other hand, since $m_\xi$ is the mode of the GEV distribution \eqref{pdfgev2},  $f_{\rm G}\big(T_{ \sigma,\delta}(x);\xi, \mu\big)$ is increasing (resp. decreasing) for each $x<x^{\min}_\xi$ (resp. $x>x^{\max}_\xi$) and $\xi>0$ (resp. $\xi<0$).
Since
$T_{ \sigma,\delta}'(x)$ is increasing and nonnegative for all $x<0$, and $f_{\rm G}\big(T_{ \sigma,\delta}(x);\xi, \mu\big)$ is  increasing  for each $x<x^{\min}_\xi$, $\xi> 0$ and $\delta<0$,  from definition \eqref{pdfbgev1} of BGEV density,
we have that $f_{\rm BGEV}(x;\xi,\mu,\sigma, \delta)=f_{\rm G}\big(T_{ \sigma,\delta}(x);\xi, \mu\big)\, T_{ \sigma,\delta}'(x)$ is the product of  two increasing and nonnegative functions.  Thus,
the BGEV PDF  is increasing for each $x<x^{\min}_\xi$. This completes the proof of first item.

The proof of the other items follows the same reasoning as the one of Item 1).
\end{proof}

\subsection{Bimodality property}
\noindent

\begin{proposition}\label{modes}
	The point $x$ is a critical point of {\rm BGEV} density \eqref{pdfbgev1}  if it is solution of the following
	equation:
	\[
	{T_{ \sigma,\delta}''(x)}
	-
	{1+\xi-\big[1+\xi (T_{ \sigma,\delta}(x)-\mu)\big]^{-1/\xi}\over 1+\xi (T_{ \sigma,\delta}(x)-\mu)}\,
	\big[T_{ \sigma,\delta}'(x)\big]^2
	=0,
	\]
	where $T_{ \sigma,\delta}(x)$, $T_{ \sigma,\delta}'(x)$ and $T_{ \sigma,\delta}''(x)$ are given in \eqref{trans} and \eqref{t-lin}.
\end{proposition}
\begin{proof}
	The proof is trivial and omitted.
\end{proof}

\begin{lemma}\label{Number of critical points}
Let $1/\xi$ be a natural number and $\delta\geqslant 2$.
The PDF of the {\rm BGEV} distribution has at most three real critical points.
\end{lemma}
\begin{proof}
Let  $y=1+\xi (T_{ \sigma,\delta}(x)-\mu)$.
Since $\xi$ is a positive number, $y> 0$. By using definitions of $T_{ \sigma,\delta}(x)$, $T_{ \sigma,\delta}'(x)$ and $T_{ \sigma,\delta}''(x)$ in \eqref{trans} and \eqref{t-lin}, a simple algebraic manipulation shows that the equation of Proposition \ref{modes} is equivalent to
\begin{align*}
\vert x\vert^{-2}T_{ \sigma,\delta}(x)
\big[\delta y^{1+1/\xi}-(\xi+1)y^{1/\xi}+1\big]=0, \quad y>0.
\end{align*}
For $\delta\geqslant 2$, $x=0$ is a critical point of $f_{\rm BGEV}(x;\xi,\mu,\sigma,\delta)$.
So the number of positive roots of the polynomial $p(y)=\delta y^{1+1/\xi}-(\xi+1)y^{1/\xi}+1$, when $1/\xi$ a natural number, determines the number of remaining roots of the BGEV PDF.

By Descartes' rule of signs (Griffiths 1947 
; Xue 2012 
 ), the polynomial $p(y)$ has two sign changes (the sequence signs is $+$, $-$, $+$), meaning that this polynomial has two or zero positive roots.
 This completes the proof.
\end{proof}
\begin{theorem}[Uni- or bimodality]
If $1/\xi\geqslant 2$ is an even natural number and $\delta\geqslant 2$, then
the PDF of the {\rm BGEV} distribution is uni- or bimodal.
\end{theorem}
\begin{proof}	
If $1/\xi$ is even, the polynomial $p(y)$ defined in Lemma \ref{Number of critical points} has odd degree.
Let $p(x)$ be the extension of polynomial $p(y)$ to $x$ in $\mathbb{R}$. Again, by Descartes' rule of signs  the polynomial $p(-x)$ has one sign change (the sequence signs is $-$, $-$, $+$), meaning that this one has one negative root. But as a sub-product of the proof of Lemma \ref{Number of critical points} we get that $p(x)$ has two or zero positive roots.
Since $p(x)$ has odd degree and the complex roots are given in pairs, we only have two possibilities: (a) $p(x)$ has  one negative root and zero positive roots or (b) $p(x)$ has one negative root and two positive roots.  From case (a) it follows that $p(y)$, $y>0$, has no positive roots, and by Lemma \ref{Number of critical points}, $x=0$ is the only critical point of BGEV density. Already, from case (b) it follows that $p(y)$, $y>0$, has exactly two positive roots, and therefore, the BGEV PDF 
has exactly three real critical points: $x=0$, $x=T_{ \sigma,\delta}^{-1}(\mu+(y_1-1)/\xi)$ and $x=T_{ \sigma,\delta}^{-1}(\mu+(y_2-1)/\xi)$.

Finally, since $\lim_{x\to\pm\infty}f_{\rm BGEV}(x;\xi,\mu,\sigma,\delta)=0$, the uni- or bimodality of the BGEV density follows.
\end{proof}

\subsection{Stochastic representation of the {\rm BGEV} distribution}
\noindent

\begin{proposition}
	Let $X\sim F_{\rm BGEV}(\cdot;\xi,\mu,\sigma, \delta)$.
	\noindent
	\begin{itemize}
		\item[\rm 1)] If $X\sim F_{\rm BGEV}(\cdot;\xi,\mu,\sigma, \delta)$, then $Y=T_{\sigma,\delta}(X)\sim F_{\rm G}(\cdot; \xi, \mu)$.
		\item[\rm 2)] If $Y\sim F_{\rm G}(\cdot; \xi, \mu)$, then $X=T_{ \sigma,\delta}^{-1}(Y)\sim F_{\rm BGEV}(\cdot; \xi,\mu,\sigma, \delta)$.
		\item[\rm 3)] If $X\sim F_{\rm BGEV}(\cdot;\xi,\mu,\sigma, \delta)$ and $c\neq 0$ a constant, then $X\sim F_{\rm BGEV}(\cdot;\xi,\vert c\vert\mu, c^2\sigma, \delta)$.
	\end{itemize}
\end{proposition}
\begin{proof}
	If $X\sim F_{\rm BGEV}(\xi,\mu,\sigma, \delta)$, by \eqref{cdfbgev1}, follows that
	\begin{eqnarray*}
		\mathbb{P}(Y\leqslant y)=\mathbb{P}(X\leqslant T_{ \sigma,\delta}^{-1}(y))
		=F_{\rm BGEV}(T_{\sigma,\delta}^{-1}(y); \xi,\mu,\sigma, \delta )
		=(F_{\rm G}\circ T_{\sigma,\delta})(T_{ \sigma,\delta}^{-1}(y);  \xi, \mu)
		=F_{\rm G}(y; \xi, \mu).
	\end{eqnarray*}
	This prove the Item 1). Analogously, the proof of second item follows.
	
	Already, the proof of Item 3) follows by combining the Item 1), the relation $\vert c\vert T_{\sigma,\delta}(x)=T_{\vert c\vert\sigma,\delta}(x)$, the following well-known fact:
	If $Y\sim F_{\rm G}(\cdot; \xi, \mu)=F_{\rm G}(\cdot; \xi, \mu,1)$, then $cY\sim F_{\rm G}(\cdot; \xi, c\mu,c)$ for  $c\neq 0$; with Item 2).
\end{proof}
\subsection{Moments}
\noindent

Here we give a closed analytical formula for the  $k$th moment  of a random variable with {\rm BGEV} distribution.
\begin{proposition}\label{prop-cit}
	If $X\sim F_{\rm BGEV}(\xi,\mu,\sigma, \delta)$ and $\xi< {1}/{k}$, then
	 \begin{align*}
	 \mathbb{E}(X^{k(\delta+1)})=
	 \begin{cases}
	 \displaystyle
	  \frac{1}{\xi^k \sigma^k} \sum_{i=0}^{k} \binom{k }{i} (\xi\mu-1)^{k-i}  \Gamma (1-\xi{i}), &  \mu-\frac{1}{\xi}>0,
	 \\[0,7cm]
	 	 \displaystyle
	\frac{(-1)^{k(\delta+1)}}{\xi^k \sigma^k} \sum_{i=0}^{k} \binom{k }{i} \left(\xi\mu-1\right)^{k-i} \gamma\big(1-\xi i ; (1-\xi \mu)^{-1/\xi}\big)
	\\
	\displaystyle
	\hspace*{0.67cm}
	+ \frac{1}{\xi^k \sigma^k} \sum_{i=0}^{k} \binom{k }{i} \left(\xi\mu-1\right)^{k-i} \Gamma\big(1-\xi i ; (1-\xi \mu)^{-1/\xi}\big), &  \mu-\frac{1}{\xi}<0,
	 \end{cases}
	 \end{align*}
	where
	$\gamma(a; x)= \int_{x}^{\infty} t^{a-1}e^{-t}\, {\rm d}t$ and
	$\Gamma(a; x)= \int_{0}^{x} t^{a-1}e^{-t}\, {\rm d}t$
	are the incomplete gamma functions.
\end{proposition}
\begin{proof}
	For $\xi>0$, by definition of expectation and the expression  \eqref{pdfgev2},  we have
	\begin{eqnarray}\label{mbgevc1}
	\mathbb{E}(X^{k(\delta+1)})&=&\int_{T_{\sigma,\delta}^{-1}(\mu-1/\xi)}^{\infty} x^{k(\delta+1)} f_{\rm BGEV}(x;\xi,\mu,\sigma, \delta) \,{\rm d}x\nonumber
	\\[0,1cm]
	&=&\int_{T_{\sigma,\delta}^{-1}(\mu-1/\xi)}^{\infty} x^{k(\delta+1)} f_{\rm G}\big(T_{ \sigma,\delta}(x);\xi, \mu \big)\, {\rm d}T_{\sigma,\delta}(x)\nonumber
	\\[0,1cm]
	&=&	\int_{\mu-1/\xi}^{\infty}  \big[T^{-1}_{\sigma, \delta}(y)\big]^{k(\delta+1)} f_{\rm G}\big(y;\xi, \mu \big)\, {\rm d}y.
	\end{eqnarray}
	When replacing the expression \eqref{tinv} of $T_{\sigma,\delta}^{-1}$ in \eqref{mbgevc1} we obtain two cases:
	 \\
	 (i) If  $\mu-{1}/{\xi}>0$,
	\begin{eqnarray*}
	\mathbb{E}(X^{k(\delta+1)})
	= \frac{1}{\sigma^k}\int_{\mu-1/\xi}^{\infty} y^{k} f_{\rm G}\big(y;\xi, \mu \big)\, {\rm d}y
	= \frac{1}{\sigma^k}\,\mathbb{E}(Y^k),  \  \  \  \ Y\sim F_{\rm G}(\cdot; \xi, \mu),
	\end{eqnarray*}
	where
		\begin{eqnarray*}
		\mathbb{E}(Y^k)=\sum_{i=1}^{k}\binom{k}{i}\Big(\mu-\frac{1}{\xi}\Big)^{k-i} \frac{1}{\xi^i} \,\Gamma(1-\xi i), \quad  \xi<1/k.
		\end{eqnarray*}
	 \\
	(ii) If  $\mu-{1}/{\xi}<0$,
\begin{eqnarray}\label{mbgevc3}	
\mathbb{E}(X^{k(\delta+1)})
&=& \frac{(-1)^{{k}{(\delta+1)}}}{\sigma^{k}}\int_{\mu-1/\xi}^{0}  y^{k} f_{\rm G}\big(y;\xi, \mu \big)\, {\rm d}y + \frac{1}{\sigma^k}\int_{0}^{\infty}  y^{k} f_{\rm G}\big(y;\xi, \mu \big)\, {\rm d}y
\nonumber
\\[0,1cm]
&=& \frac{(-1)^{{k}{(\delta+1)}}}{\sigma^{k}}\, {I_1}+ \frac{1}{\sigma^k}\, {I_2}.
\end{eqnarray}
	
	Using the  PDF \eqref{pdfgev2} and the  substitution $w^{-1}=1+\xi (y-\mu)$, we obtain
	\begin{eqnarray}\label{I1c1}
	I_1	&=& \frac{1}{\xi^{k+1}} \int_{(1-\xi \mu)^{-1}}^{\infty} [w^{-1}+(\xi \mu-1)]^{k} \  w^{\frac{1}{\xi}-1} {\rm e}^{-w^{1/\xi}}\, {\rm d}w.
	\end{eqnarray}
	To solve the integral of \eqref{I1c1} we  use the  Newton's binomial formula, thus
	\begin{eqnarray}\label{I1c2}
	I_1&=&\frac{1}{\xi^{k+1}}\sum_{i=1}^{k}\binom{k}{i} \left(\xi\mu-1\right)^{k-i} \int_{(1-\xi \mu)^{-1}}^{\infty}  w^{\frac{1}{\xi}-i-1} {\rm e}^{-w^{1/\xi}}\, {\rm d}w
	\nonumber
	\\[0,1cm]
	&=& \frac{1}{\xi^{k}}\sum_{i=1}^{k}\binom{k}{i} \left(\xi\mu-1\right)^{k-i} \int_{(1-\xi \mu)^{-1/\xi}}^{\infty}  z^{-\xi i} {\rm e}^{-z}\, {\rm d}z,
		\end{eqnarray}
	where in the last line we used the new substitution $z=w^{1/\xi}$.	
		
	Similarly we obtain
	\begin{eqnarray}\label{I2}
	I_2&=&
	\frac{1}{\xi^{k}}\sum_{i=1}^{k}\binom{k}{i}
	\left(\xi\mu-1\right)^{k-i}
		\int_{0}^{(1-\xi \mu)^{-1/\xi}}  z^{-\xi i} {\rm e}^{-z}\, {\rm d}z.
	\end{eqnarray}
	
	The proof is completed by expressing as integrals of \eqref{I1c2} and \eqref{I2} in terms of the incomplete gamma functions and then updating the equation \eqref{mbgevc3}.
	
	For $ \xi <0 $ the proof is similar.
\end{proof}

\begin{corollary}
	Let 	 $X\sim F_{\rm BGEV}(\xi,\mu,\sigma, \delta)$ and $\delta \leq 0$, then
\begin{align*}
\mathbb{E}(X)=
\begin{cases}
\displaystyle
\frac{1}{(\xi \sigma)^{ \llbracket 1/(\delta+1) \rrbracket}} \sum_{i=0}^{\llbracket 1/(\delta+1) \rrbracket} \binom{ \llbracket 1/(\delta+1) \rrbracket }{i} \left(\xi\mu-1\right)^{\llbracket 1/(\delta+1) \rrbracket-i} \Gamma (1-\xi{i}), & \!\! \mu-\frac{1}{\xi}>0,
\\[0,7cm]
\displaystyle
\frac{(-1)^{k(\delta+1)}}{(\xi \sigma)^{\llbracket 1/(\delta+1) \rrbracket}} \sum_{i=0}^{\llbracket 1/(\delta+1) \rrbracket} \binom{\llbracket 1/(\delta+1) \rrbracket }{i} \left(\xi\mu-1\right)^{\llbracket 1/(\delta+1) \rrbracket-i} \gamma\big(1-\xi i ; (1-\xi \mu)^{-1/\xi}\big)
\\
\displaystyle
+ \frac{1}{(\xi \sigma)^{\llbracket 1/(\delta+1) \rrbracket}} \sum_{i=0}^{\llbracket 1/(\delta+1) \rrbracket} \binom{\llbracket 1/(\delta+1) \rrbracket }{i} \left(\xi\mu-1\right)^{\llbracket 1/(\delta+1) \rrbracket-i}  \Gamma\big(1-\xi i ; (1-\xi \mu)^{-1/\xi}\big), & \!\! \mu-\frac{1}{\xi}<0.
\end{cases}
\end{align*}

\end{corollary}
\begin{proof} The result follows directly of Proposition \ref{prop-cit}   by considering
	$\delta=({1}/{k})-1$ and $k \in \mathbb{Z}^{+}$.
\end{proof}

%

\subsection{Quantiles}
\noindent

In the next section,  through Monte Carlo simulation studies, we test the behavior of the   parameter estimates  of the  BGEV model. The random samples of the population are simulated using the  inverse transformation method  based on quantiles.
The quantile function of a random variable $X\sim F_{\rm BGEV}(\cdot; \xi,\mu,\sigma, \delta)$ is defined by

\begin{equation*}
x_Q= \inf \{x \in \mathbb{R}: Q \leq F_{\rm BGEV}(x; \xi,\mu,\sigma, \delta) \}
\end{equation*}
or equivalent by
\begin{eqnarray*}\label{Q}
	x_Q&=&T_{\sigma,\delta}^{-1}(F_{\rm G}^{-1}(Q)).
\end{eqnarray*}
From \eqref{cdfbgev1}, we have
$F_{\rm G}^{-1}(y)=\mu+{[(-\ln y)^{-\xi}-1]}/{\xi}$. Now, by using the expression \eqref{tinv} of
$T_{\sigma,\delta}^{- 1}$, we get
\begin{eqnarray}\label{Q}
x_Q=
\begin{cases}
\displaystyle
\biggl[\frac{\mu}{\sigma} + \frac{(-\ln Q)^{-\xi}-1}{\sigma \xi}\biggr]^{{1}/{(\delta+1)}}	, & \displaystyle  \mu+  \frac{(-\ln Q)^{-\xi}-1}{ \xi}>0,
\\[0,4cm](-1)^{\frac{2+\delta}{1+\delta}}
\displaystyle
\biggl[\frac{\mu}{\sigma } + \frac{(-\ln Q)^{-\xi}-1}{\sigma \xi}\biggr]^{{1}/{(\delta+1)}}	, &  \displaystyle
\mu+  \frac{(-\ln Q)^{-\xi}-1}{ \xi}<0.
\end{cases}
\end{eqnarray}

\subsection{Tail behavior of BGEV}
\noindent

According to Embrechts et al. (2013) 
, asymptotic estimates of probability models are omnipresent in statistics, insurance mathematics, mathematical finance, and hydrology.
In this sense, the theory of regular variation plays a crucial role. Concepts and properties of regularly varying  can be found  in the  book of Bingham et al. (1987) 
In order to study  the asymptotic behavior of the tail of the BGEV distribution, here we show some of the results of  regular  variation.
\begin{definition} A positive Lebesgue measurable function $h$ on $(0, \infty)$ is regularly varying at $\infty$ of index $\alpha \in \mathbb{R}$ ($h \in \mathcal{R}_{\alpha}$)  if
\[
\lim\limits_{t \rightarrow \infty} \frac{h(tx)}{h(t)}=x^{\alpha}.
\]
\end{definition}
Suppose $F$ is a  distribution function such that $F(x)<1$ for all $x>0$.
Let $\overline{F}(x)=1-F(x)$ and $\overline{F}\in \mathcal{R}_{-\alpha} $ for some $\alpha>0$ , then
\begin{eqnarray}\label{vr}
\overline{F}(x)=o(x^{-\alpha}L(x))
\end{eqnarray}
for some $L \in \mathcal{R}_{0}$, $L_0$  a  slowly varying function.
The expression \eqref{vr} indicates that the tail of the distribution  $ F $ slowly decays, as $ x^{-\alpha} $, that is $ F $ has  heavy (right) tail  with  index $ \alpha $.
Thus, $F$ is Pareto-like distribution.

\begin{proposition}
	Let $X\sim F_{\rm BGEV}(\cdot; \xi,\mu,\sigma, \delta)$  with $\xi>0$ and $\delta>-1$ . Then
	\begin{eqnarray}\label{tail}
	\overline{F}_{\rm BGEV}(x; \xi,\mu,\sigma, \delta)=o\left(  x^{- \frac{\delta +1}{\xi}}\right), \ \ x>0.
	\end{eqnarray}
	\begin{proof}
		The proof follows of the straightforward calculation of
		\[
		\lim\limits_{t \rightarrow \infty}
		\frac{	\overline{F}_{\rm BGEV}(tx; \xi,\mu,\sigma, \delta)}
		{\overline{F}_{\rm BGEV}(t; \xi,\mu,\sigma, \delta)}
		=x^{- \frac{\delta +1}{\xi}},
		\]
		where ${F}_{\rm BGEV}(tx; \xi,\mu,\sigma, \delta)$ is given in \eqref{cdfbgev1}.
	\end{proof}
\end{proposition}
From \eqref{tail} we have that ${F}_{\rm BGEV}$ is Pareto-like distribution with tail index ${(\delta +1)}/{\xi}$. The new $  \delta $ parameter affects the weight of the tail.

\section{ Graphical study }
\label{sec:4}
\noindent

The role that the parameters $ \xi $, $ \mu $, $ \sigma $, $ \delta $, and  play in the BGEV distribution were investigated graphically by generating different densities with variations in each parameter. Fig. \ref{F1} shows that  $ \sigma $ is a  scale parameter. According to Fig. \ref {F2},
the shape of the BGEV PDF changed by varying the $\mu $ parameter.  That is, $ \mu $ is a shape parameter. Fig. \ref{F2}   suggests that  depending
on the combinatation of the parameters values $\xi$ and $\mu$ the BGEV density can be unimodal or bimodal. For example,
$ f_{\rm BGEV} (.; 2, \mu, 1, 2) $  is unimodal and $f_{{\rm BGEV}} (.; 0.5, \mu, 1, 2) $  is bimodal.
In the bimodal case, the first mode value is greater than the  second mode  value for $\mu \in [0, 0.25)$, while for  $ \mu > 0.25$ the second mode value is greater than the first mode.
Figures \ref{F3} and \ref{F4} correspond to variations of the new parameter
$ \delta $. When $\delta =0 $ the BGEV coincides with the basic GEV  distribution. Two scenarios were considered. Scenario 1 for  $\xi<0$ (see Fig. \ref{F3}), and  scenario 2 for  $\xi>0$ (see Fig. \ref{F4} ).
In scenario 1 , Fig. \ref{F3}, when $ \delta \leq 0 $ the PDF is unimodal and for $ \delta > 0 $ the PDF is bimodal. In bimodal case, the mode values increases as delta grows.
In scenario 2, the Fig. \ref{F4}  shows  that when delta takes positive or negative values close to zero and  $\xi>1$ the PDF is bimodal. In this case, the mode value also increases as delta increases.
%

\begin{figure}[!htbp]
	\caption{ } 	
	\begin{center}
		\vspace{-0.8cm}	
		\includegraphics[width=0.47\linewidth]{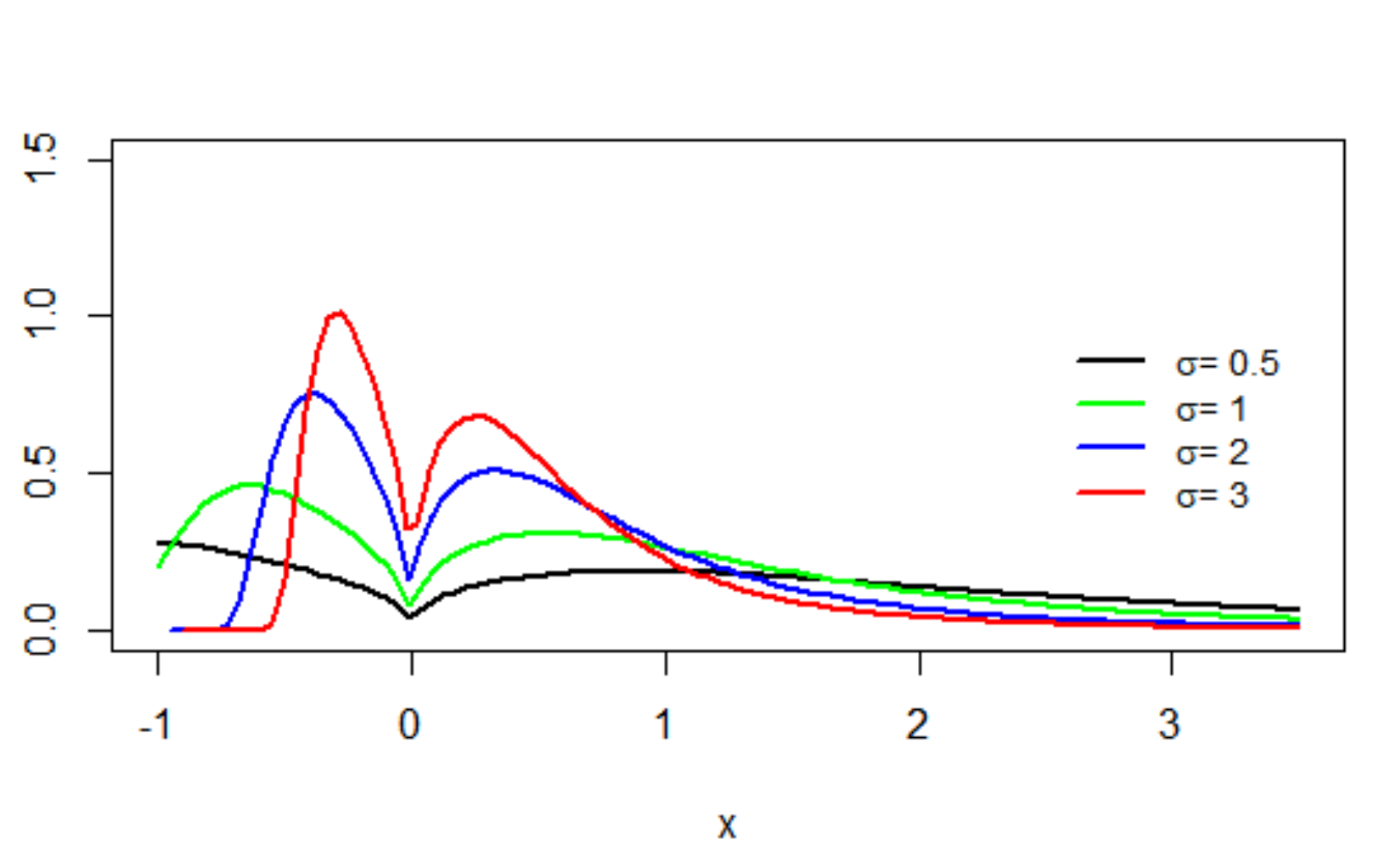}
		\includegraphics[width=0.47\linewidth]{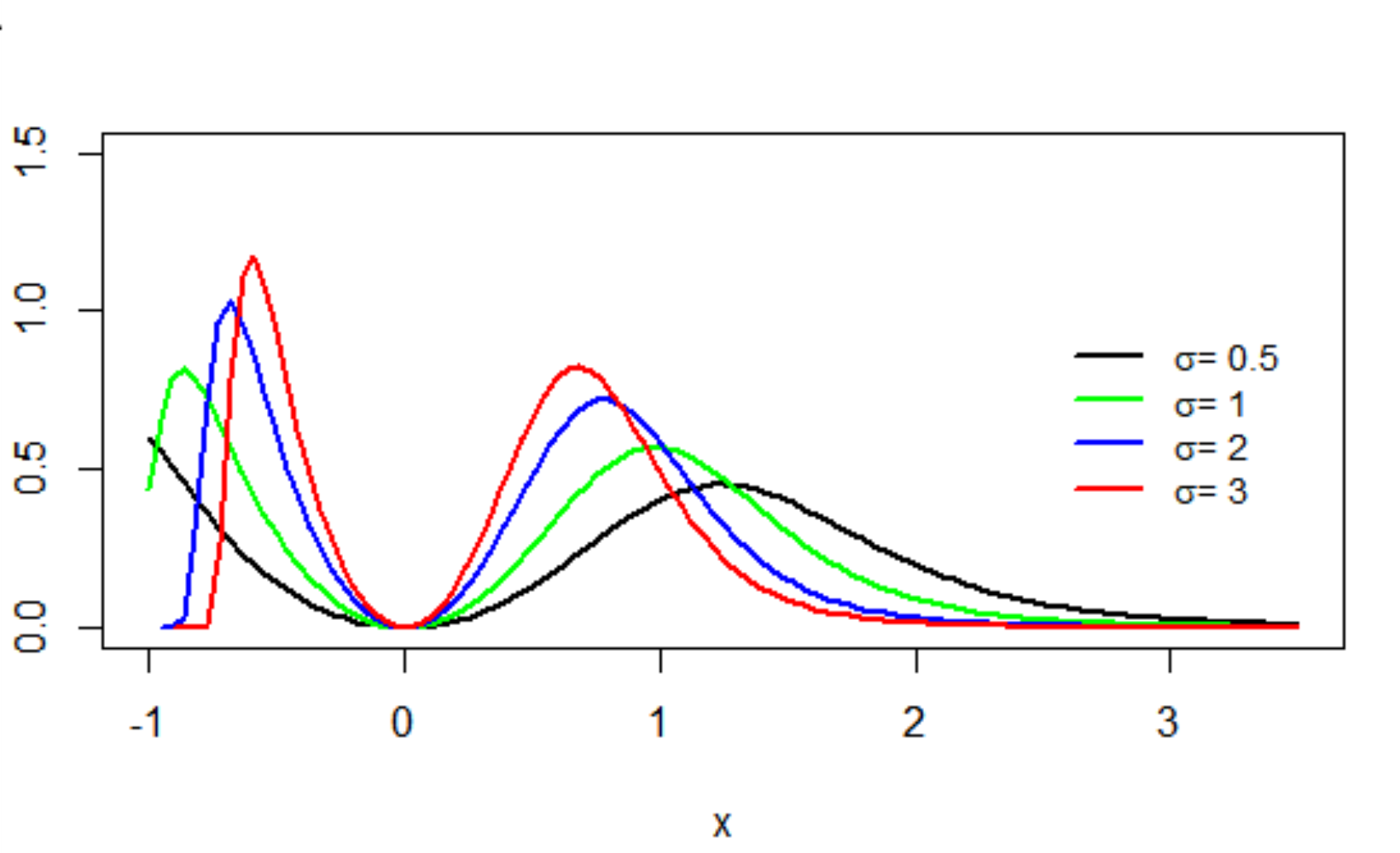}
			\vspace{-0.2cm}
	\end{center} 
	{Left: BGEV PDF, $f_{\rm BGEV}(x; 0.5, 0, \sigma,0.4)$, with $\sigma$ varying; Right: BGEV PDF, $f_{\rm BGEV}(x; 0.5, 0, \sigma,2)$ with  $\sigma$ varying.}
	\label{F1}
\end{figure}
%
\begin{figure}[!htbp]
	\caption{} 	
	\begin{center}
		\vspace{-1cm}	 	
		\includegraphics[width=0.47\linewidth]{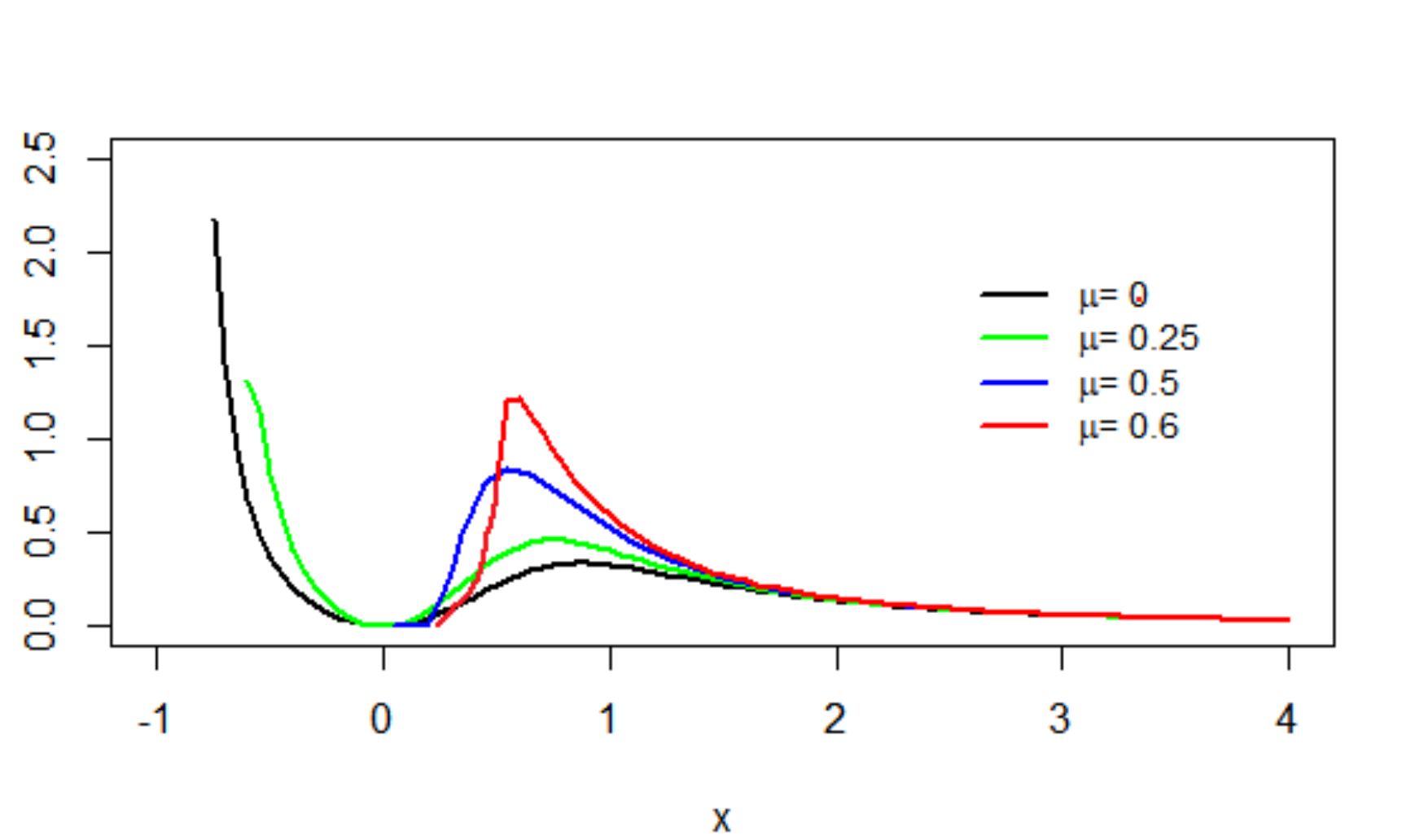}
		\includegraphics[width=0.47\linewidth]{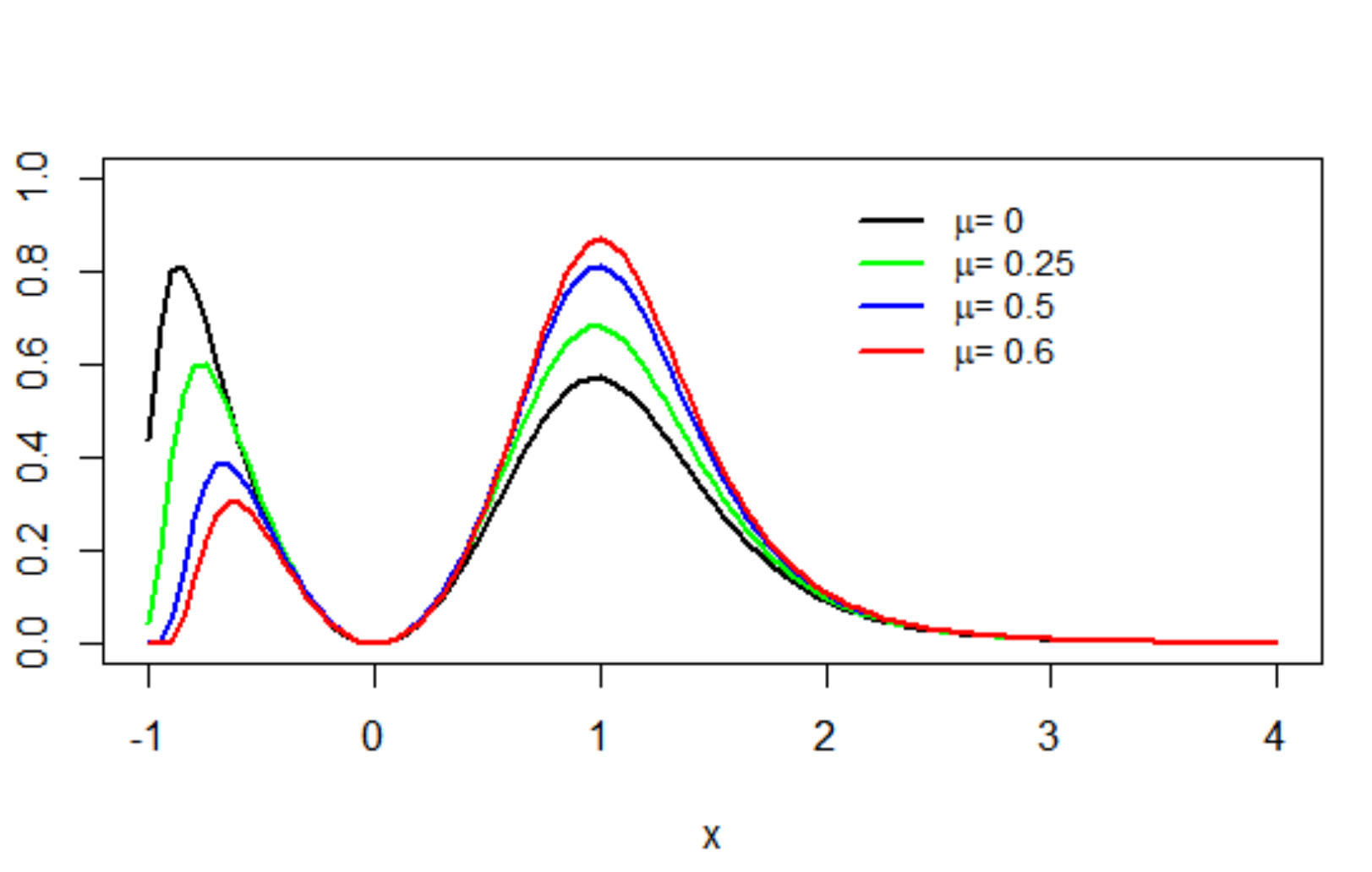}
		\end{center}
\vspace{-0.2cm}
	{Left: BGEV PDF, $f_{\rm BGEV}(x; 2, \mu, 1, 2)$, with  $\mu$ varying;  Right: BGEV PDF, $f_{\rm BGEV}(x; 0.5, \mu, 1, 2)$, with  $\mu$ varying.}
	\vspace{-2.5cm}
	\label{F2}
\end{figure}
%
%
\begin{figure}[!htbp]
	\caption{} 	
\begin{center}
		\vspace{-1.09cm}	 	
		\includegraphics[width=0.47\linewidth]{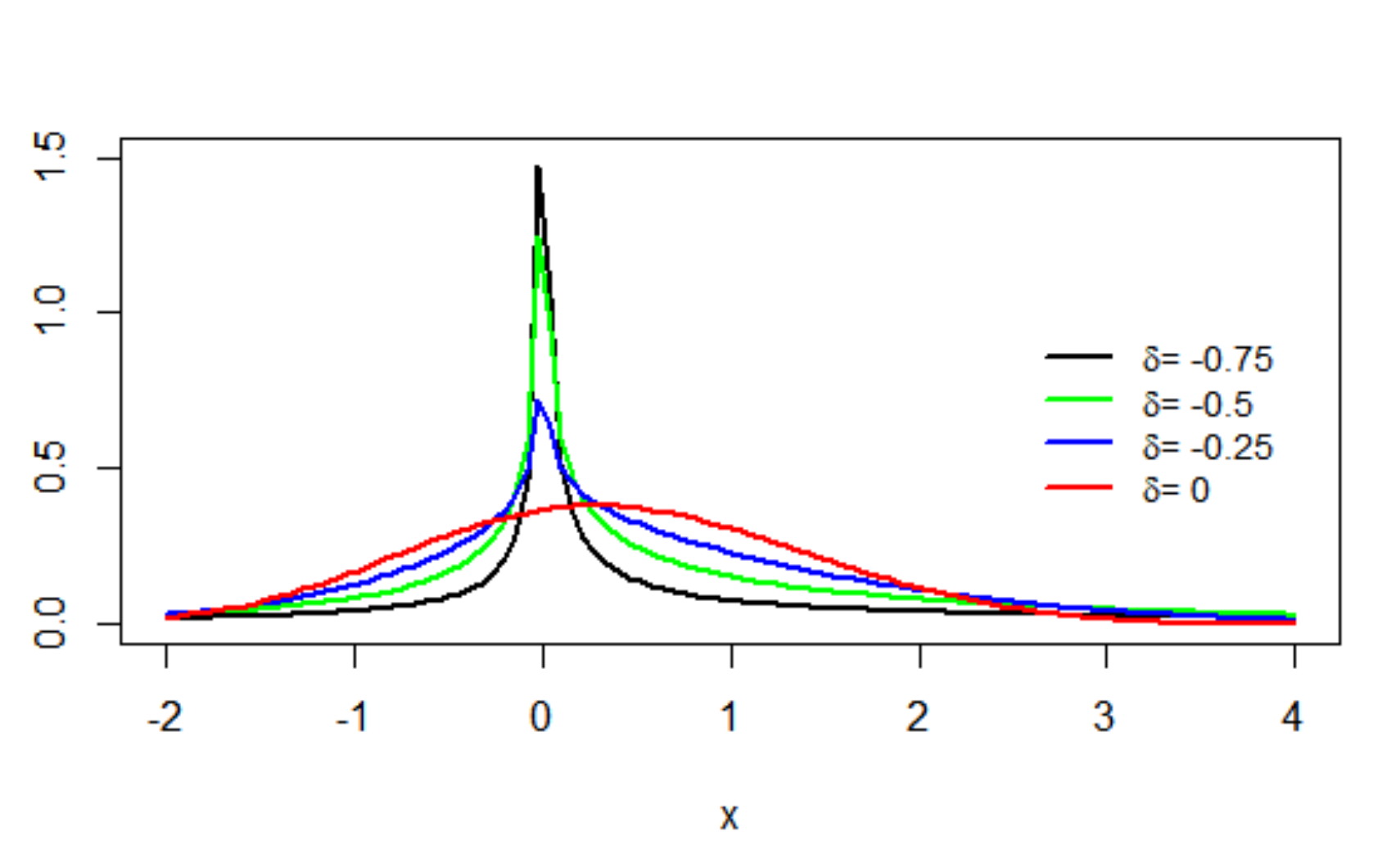}
		\includegraphics[width=0.47\linewidth]{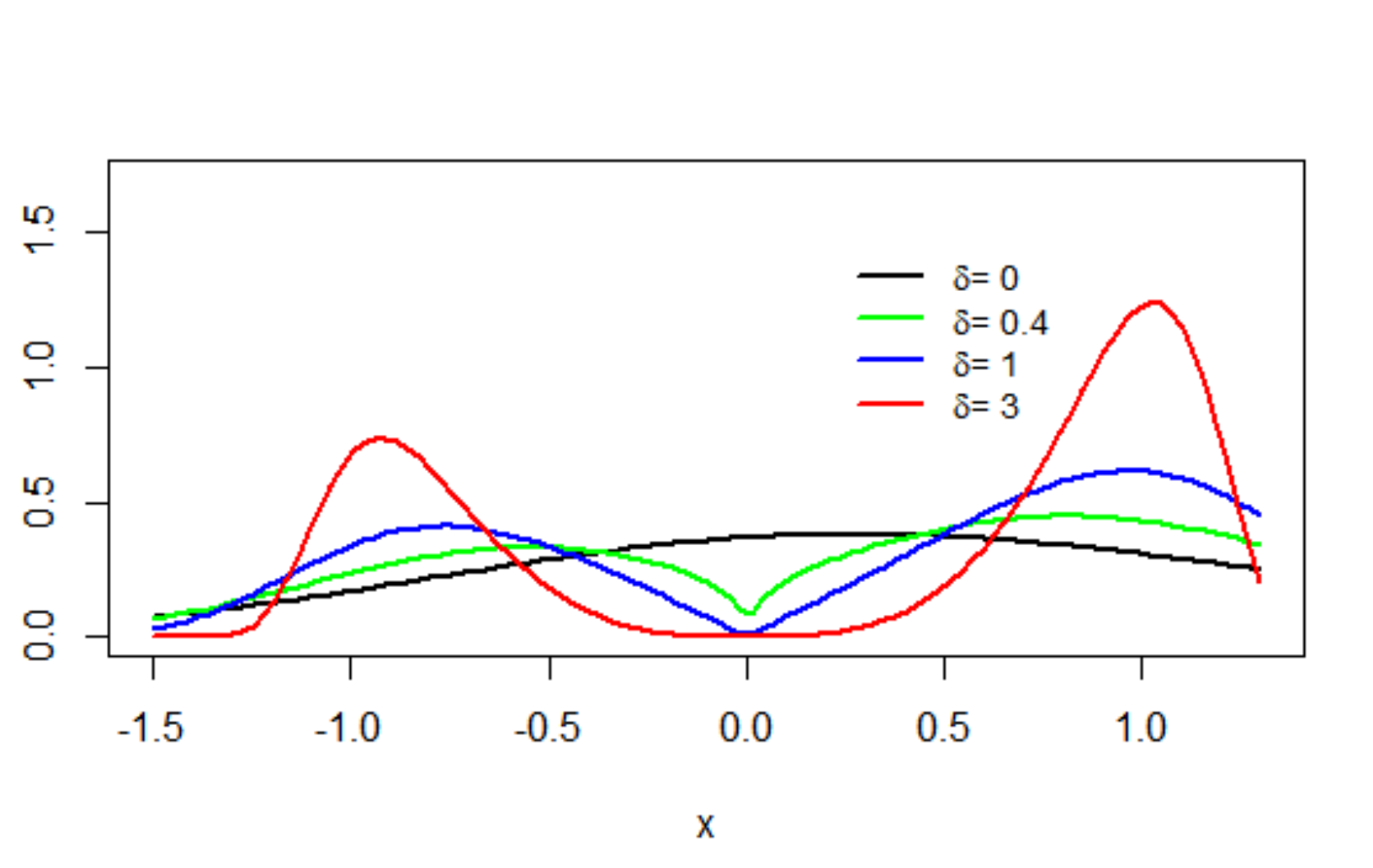}
		\includegraphics[width=0.47\linewidth]{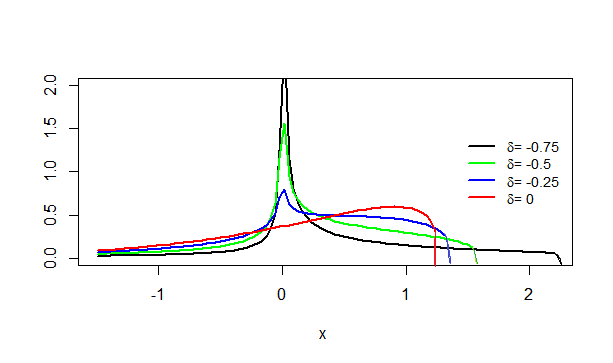}
		\includegraphics[width=0.47\linewidth]{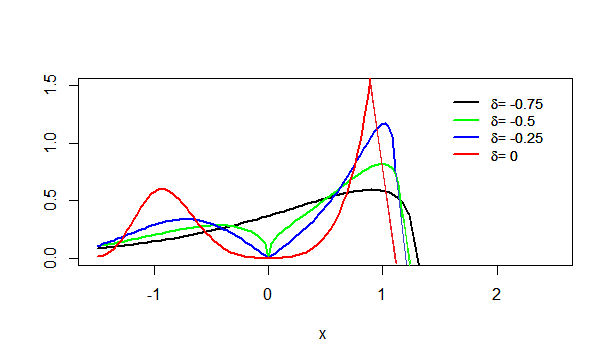}
		\vspace{-0.7cm}
\end{center}
	{Top left and right: BGEV PDF, $f_{\rm BGEV}(x; -0.25, 0, 1, \delta)$ with  $\delta$ varying;  Bottom left and right: BGEV PDF, $f_{\rm BGEV}(x; -1, 0, 1, \delta)$ with  $\delta$ varying.}
	\vspace{-0.5cm}
	\label{F3}
\end{figure}
%

\begin{figure}[!htbp]
	\caption{} 	
\begin{center}
		\vspace{-0.9cm}	 	
		\includegraphics[width=0.47\linewidth]{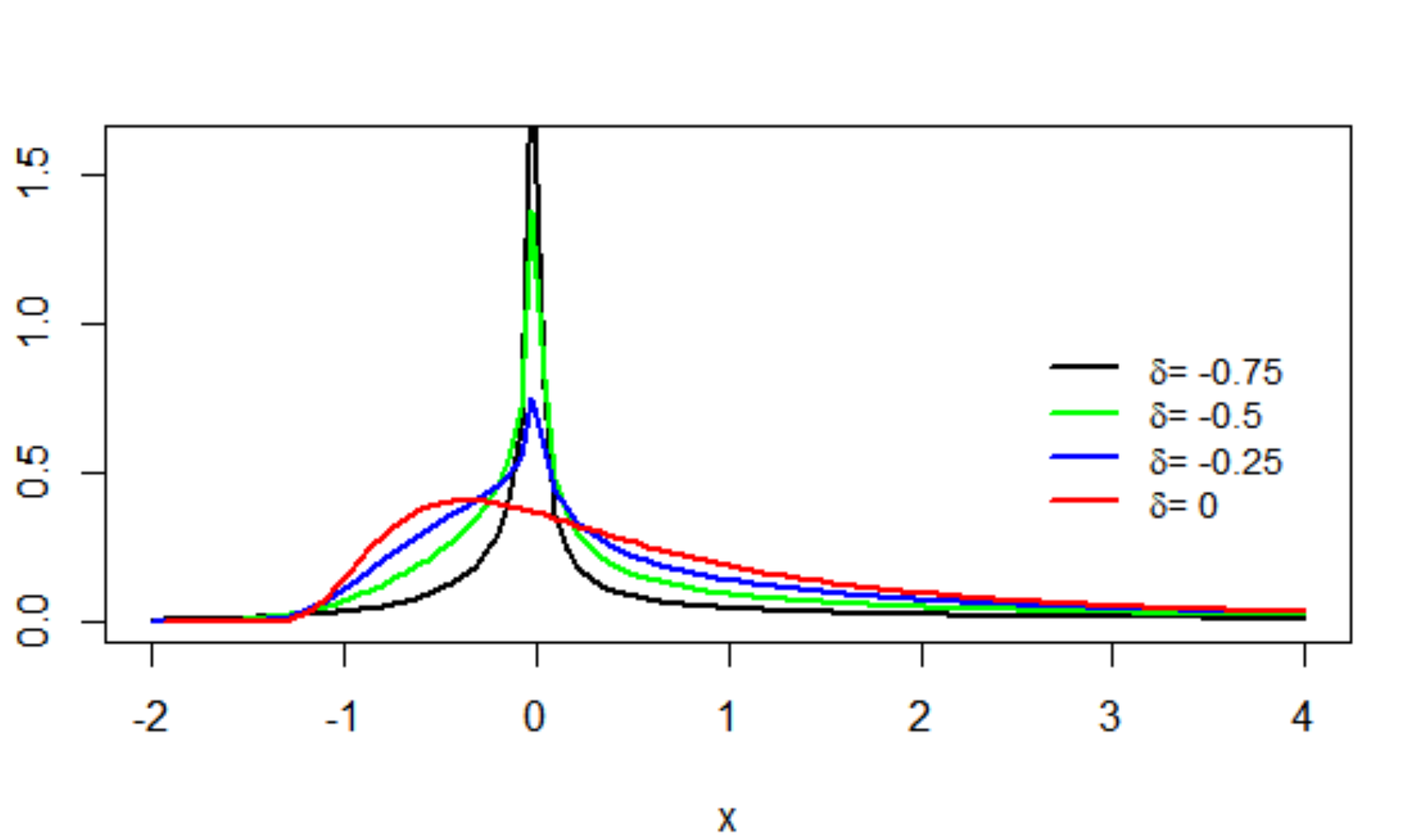}
		\includegraphics[width=0.47\linewidth]{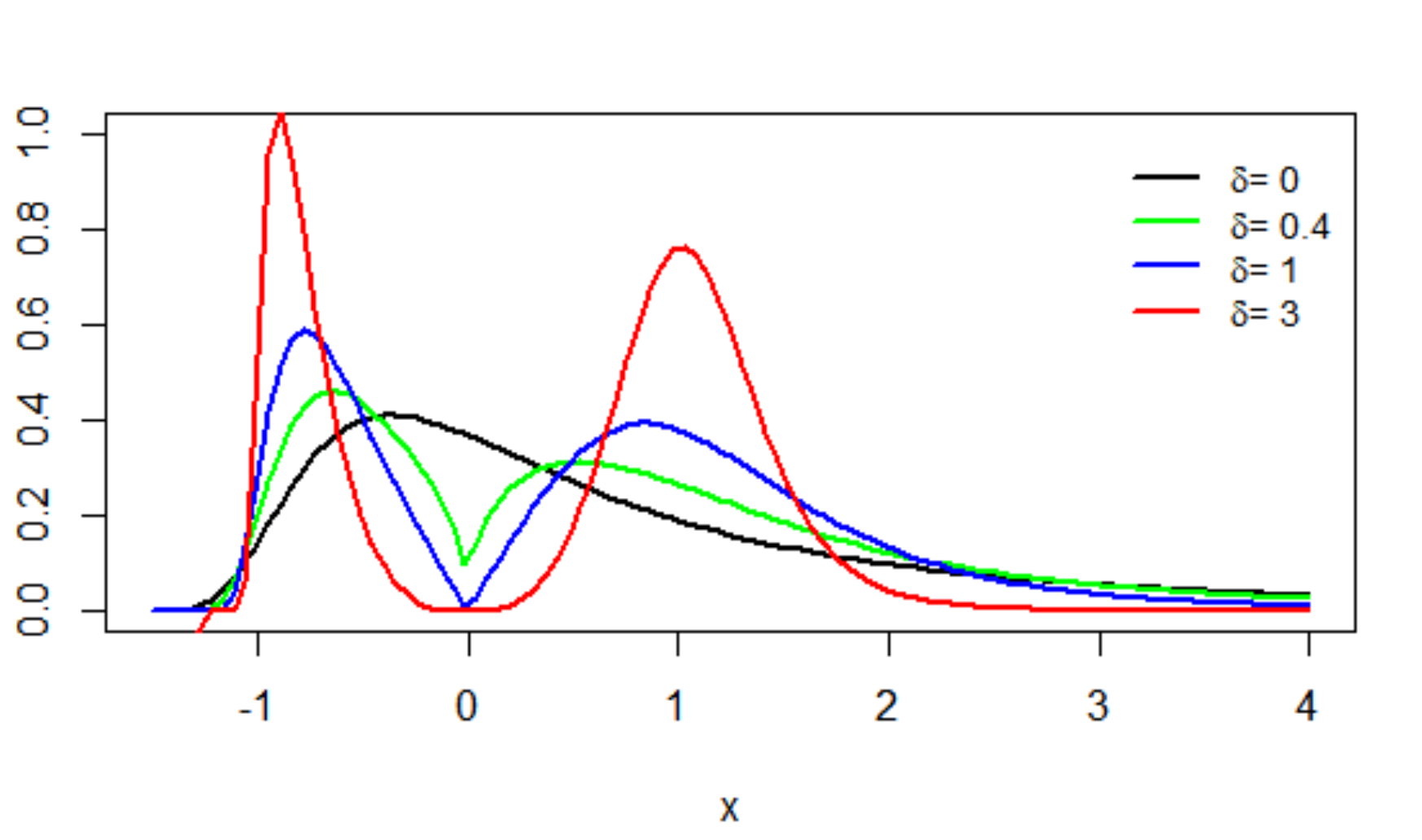}
		\includegraphics[width=0.47\linewidth]{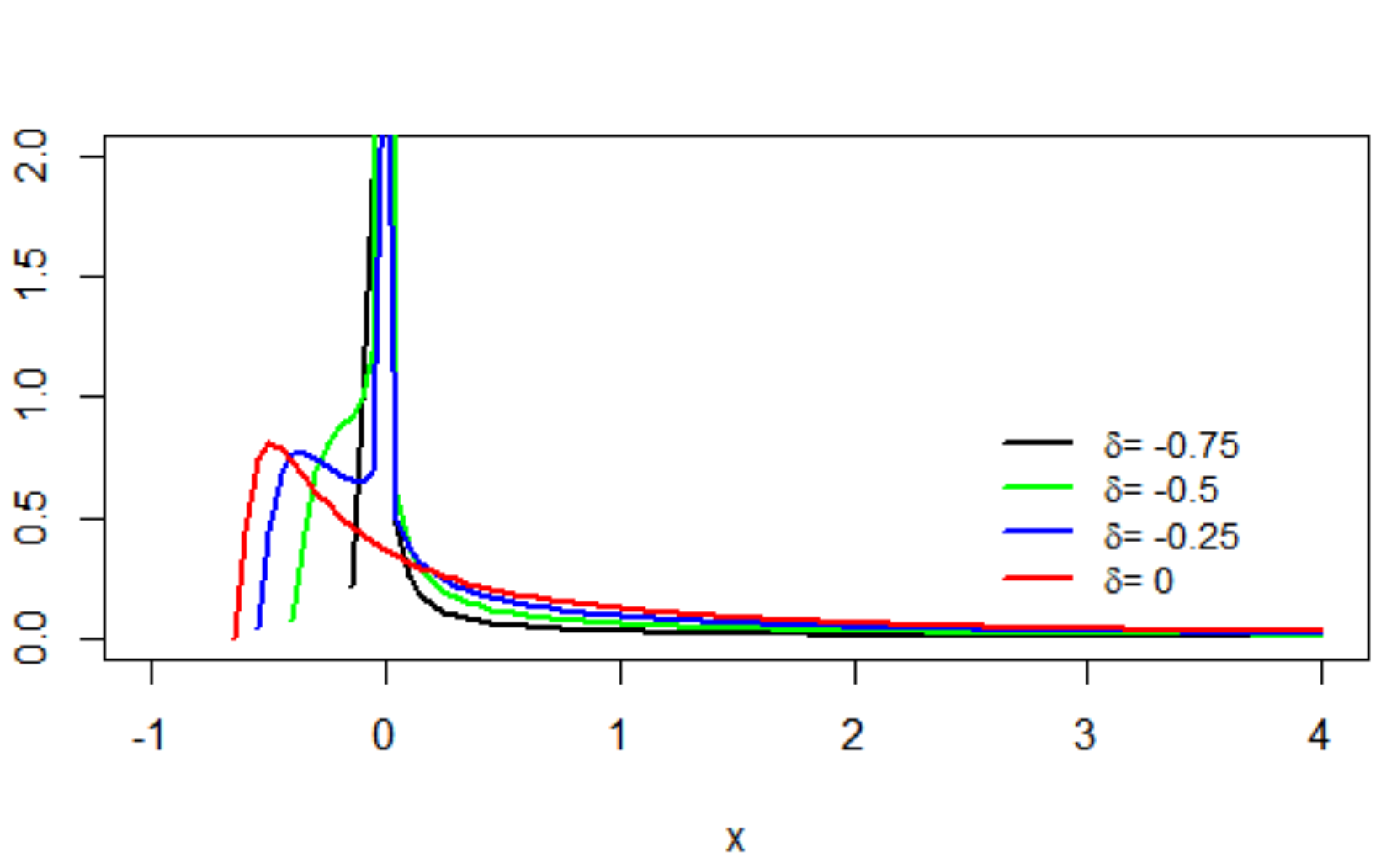}
		\includegraphics[width=0.47\linewidth]{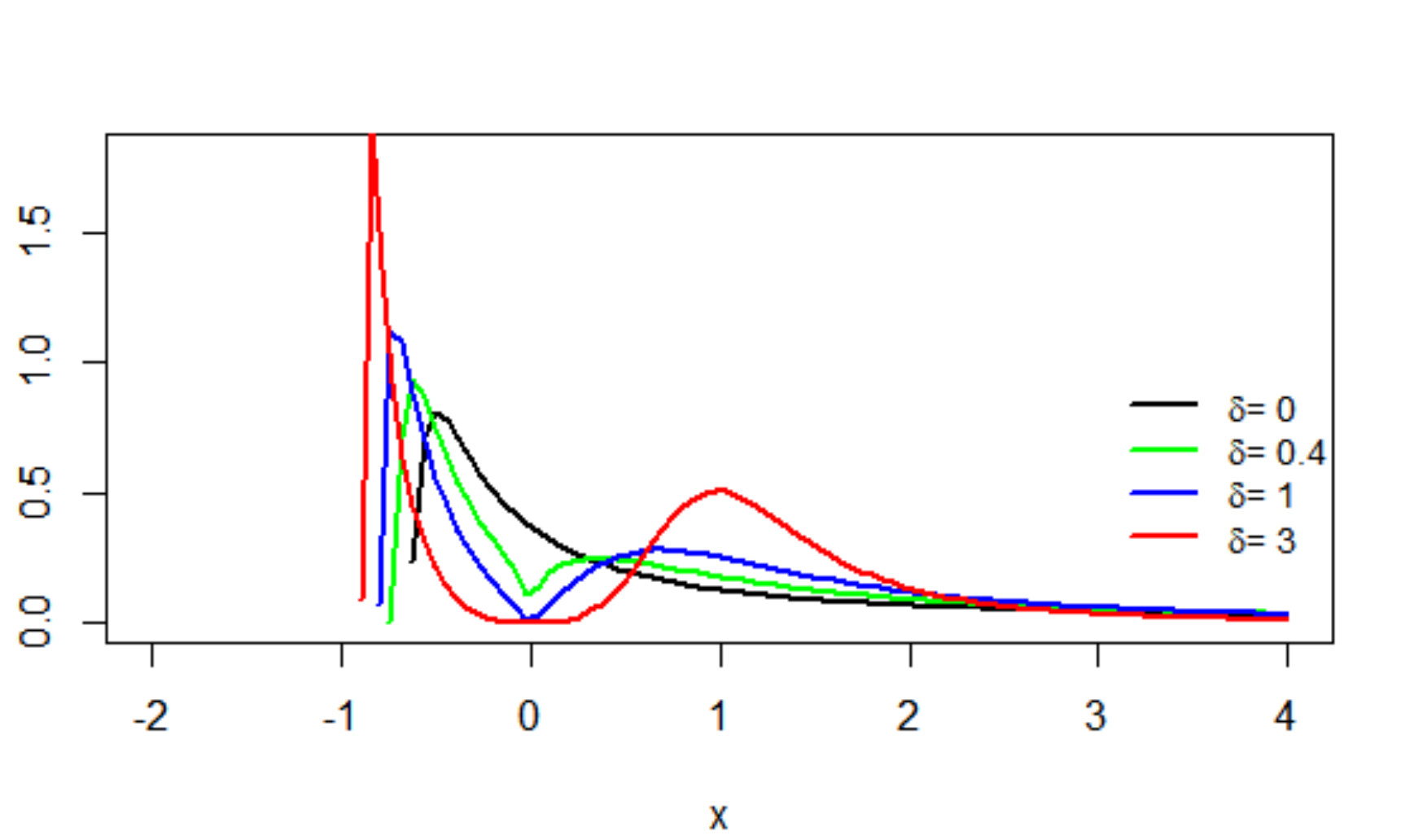}
\end{center}
	{Top left and right: BGEV PDF, $f_{\rm BGEV}(x; 0.5, 0, 1, \delta)$ with  $\delta$ varying; Bottom left and right: BGEV PDF, $f_{\rm BGEV}(x; 1.5, 0, 1, \delta)$ with  $\delta$ varying.}
	\vspace{-2.5cm}
	\label{F4}
\end{figure}

\newpage
\section{ Maximum likelihood estimation}\label{mle}
\noindent
\label{sec:5}

In this section we present the expressions to be calculated to obtain the maximum likelihood estimators for the parameters of a BGEV distribution.

Let $ X \sim F_{\rm BGEV} $ be a random variable with PDF $ f _{\rm BGEV} (x; \Theta) $ defined in \eqref{pdfbgev1}, where $\Theta=(\xi,\mu,\sigma, \delta)$. Let $ (X_1, \dots, X_n) $ be a random sample of X and
$x= (x_1, \dots, x_n) $ the corresponding observed values of the random sample $ (X_1, \dots, X_n) $.
The  log-likelihood function  for $\Theta=(\xi,\mu,\sigma, \delta)$ is given by
\begin{align}\label{fv1}
\ell_n(\Theta;x)
=
n\ln\sigma+n\ln(\delta+1)
+\sum_{i=1}^{n}
\Big[
\delta \ln\vert x_i\vert-\Big(1+{1\over \xi}\Big)
\ln\Psi_i(\Theta)
-
\Psi_i^{-1/\xi}(\Theta)
\Big],
\end{align}
where
$
\Psi_i(\Theta)=1+\xi\big(\sigma x_i|x_i|^{\delta}-\mu\big),\ i=1, \ldots, n.
$
Note that, for all $i=1,\ldots,n$,
\begin{align*}
\begin{array}{lllll}
&\frac{\partial \Psi_i(\Theta)}{\partial \mu}
=
-\xi,
&\frac{\partial \Psi_i(\Theta)}{\partial \sigma}
=
\xi x_i|x_i|^{\delta},
\\[0,2cm]
&\frac{\partial \Psi_i(\Theta)}{\partial \delta}
=
\xi\sigma x_i|x_i|^{\delta} \ln|x_i|,
&\frac{\partial \Psi_i(\Theta)}{\partial \xi}
=
\sigma x_i|x_i|^{\delta}-\mu.
\end{array}
\end{align*}
Then, the maximum likelihood  estimates of $\mu,\sigma,  \delta, \xi $ are the solutions of the following system of equations
\begin{eqnarray} \label{xi}
\begin{array}{llllll}
\frac{\partial \ell_n(\Theta;x)}{\partial \mu}
&=&
\mu
\sum_{i=1}^{n}
\Omega_i(\Theta)=0,
\\[0,2cm]
\frac{\partial \ell_n(\Theta;x)}{\partial \sigma}
&=&
\frac{n}{\sigma}
-
\sum_{i=1}^{n}
x_i|x_i|^{\delta}
\Omega_i(\Theta)
=0,
\\[0,2cm]
\frac{\partial \ell_n(\Theta;x)}{\partial \delta}
&=&
\frac{n}{\delta+1}
+
\sum_{i=1}^{n}
[
\ln|x_i|
+	
\sigma x_i|x_i|^{\delta} \ln|x_i| \,
\Omega_i(\Theta)
]
=0,
\\[0,2cm]
\frac{\partial \ell_n(\Theta;x)}{\partial \xi}
&=&
\xi^{-2}
\sum_{i=1}^{n}
(\sigma x_i|x_i|^{\delta}-\mu)
\left\{
\ln\Psi_i(\Theta)
 -
\xi
\Omega_i(\Theta)
\right\}
=0,
\end{array}
\end{eqnarray}
where
$
\Omega_i(\Theta)
=\Psi_i^{-1}(\Theta)
[
1+\xi
-
\Psi_i^{-{1}/{\xi}}(\Theta)
],\ i=1,\ldots,n.
$
Since
\begin{align*}
\begin{array}{lllll}
{\partial\Omega_i(\Theta)\over \partial \theta}
=
-\xi^{-1} \Psi_i^{-1}(\Theta) \big[
\xi \Omega_i^{-1}(\Theta) - \Psi_i^{-1/\xi}(\Theta)
\big]
\frac{\partial \Psi_i(\Theta)}{\partial \theta}, \quad \theta\in \{\xi,\mu,\sigma, \delta\},
\end{array}
\end{align*}
the second-order partial derivatives of the log-likelihood function  $\ell_n(\Theta;x)$ are given by
%
\begin{align}\label{sec-der-1}
\begin{array}{lllll}
\frac{\partial^2 \ell_n(\Theta;x)}{\partial \mu^2}
=
\sum_{i=1}^{n}
\{
\Omega_i(\Theta)
+
\mu
\Psi^{-1}_i(\Theta)
[\xi \Omega_i(\Theta)-\Psi_i^{-{1}/{\xi}}(\Theta)]
\},
\\[0,2cm]
\frac{\partial^2 \ell_n(\Theta;x)}{\partial \sigma^2}
=
-\frac{n}{\sigma^2}
+
\sum_{i=1}^{n}
x_i^2|x_i|^{2\delta}
\Psi^{-1}_i(\Theta)
[\xi \Omega_i(\Theta)-\Psi_i^{-{1}/{\xi}}(\Theta)],
\\[0,2cm]
\frac{\partial^2 \ell_n(\Theta;x)}{\partial \delta^2}
=
-\frac{n}{(\delta+1)^2}
+
\sigma
\sum_{i=1}^{n}
x_i|x_i|^{2\delta}\left(\ln|x_i|\right)^2
\big\{
\Omega_i(\Theta)
-
\sigma x_i
\Psi^{-1}_i(\Theta)
[\xi \Omega_i(\Theta)-\Psi_i^{-{1}/{\xi}}(\Theta)]
\big\},
\\[0,2cm]
\frac{\partial^2 \ell_n(\Theta;x)}{\partial \xi^2}
=
-2\xi^{-3}
\sum_{i=1}^{n}
(\sigma x_i|x_i|^{\delta}-\mu)
[
\ln\Psi_i(\Theta)
 -
\xi \Omega_i(\Theta)
]
\\[0,2cm] \qquad\qquad\quad\
-
\xi^{-2}
\sum_{i=1}^{n}
(\sigma x_i|x_i|^{\delta}-\mu)^2
\big\{
\Omega_i(\Theta) -
\Psi^{-1}_i(\Theta)
[\xi \Omega_i(\Theta)-\Psi_i^{-{1}/{\xi}}(\Theta)+1]
\big\},
\end{array}
\end{align}
and, by Schwarz's theorem or Clairaut's theorem on equality of mixed partials (James 1966)
, the second-order mixed derivatives of $\ell_n(\Theta;x)$ can be written as
\begin{align}\label{sec-der-2}
\begin{array}{lllll}
\frac{\partial^2 \ell_n(\Theta;x)}{\partial \mu\partial \sigma}
=
\frac{\partial^2 \ell_n(\Theta;x)}{\partial \sigma\partial \mu}
=
-
\mu
\sum_{i=1}^{n}
x_i|x_i|^{\delta}
\Psi^{-1}_i(\Theta)
[\xi \Omega_i(\Theta)-\Psi_i^{-{1}/{\xi}}(\Theta)],
\\[0,2cm]
\frac{\partial^2 \ell_n(\Theta;x)}{\partial \mu\partial \delta}
=
\frac{\partial^2 \ell_n(\Theta;x)}{\partial \delta\partial \mu}
=
-
\mu\sigma
\sum_{i=1}^{n}
x_i|x_i|^{\delta} \ln |x_i|
\Psi^{-1}_i(\Theta)
[\xi \Omega_i(\Theta)-\Psi_i^{-{1}/{\xi}}(\Theta)],
\\[0,2cm]
\frac{\partial^2 \ell_n(\Theta;x)}{\partial \mu\partial \xi}
=
\frac{\partial^2 \ell_n(\Theta;x)}{\partial \xi\partial \mu}
=
-\mu\xi^{-1}
\sum_{i=1}^{n}
(\sigma x_i|x_i|^{\delta}-\mu)
\Psi^{-1}_i(\Theta)
[\xi \Omega_i(\Theta)-\Psi_i^{-{1}/{\xi}}(\Theta)],
\\[0,2cm]
\frac{\partial^2 \ell_n(\Theta;x)}{\partial \sigma\partial \delta}
=
\frac{\partial^2 \ell_n(\Theta;x)}{\partial \delta\partial \sigma}
=
-
\sum_{i=1}^{n}
x_i|x_i|^{2\delta} \ln |x_i|
\big\{
\Omega_i(\Theta)
-
\sigma x_i
\Psi^{-1}_i(\Theta)
[\xi \Omega_i(\Theta)-\Psi_i^{-{1}/{\xi}}(\Theta)]
\big\},
\\[0,2cm]
\frac{\partial^2 \ell_n(\Theta;x)}{\partial \sigma\partial \xi}
=
\frac{\partial^2 \ell_n(\Theta;x)}{\partial \xi\partial \sigma}
=
-\xi^{-1}
\sum_{i=1}^{n}
x_i|x_i|^{\delta} \big(\sigma x_i|x_i|^{\delta}-\mu\big)
\Psi^{-1}_i(\Theta)
[\xi \Omega_i(\Theta)-\Psi_i^{-{1}/{\xi}}(\Theta)],
\\[0,2cm]
\frac{\partial^2 \ell_n(\Theta;x)}{\partial \delta\partial \xi}
=
\frac{\partial^2 \ell_n(\Theta;x)}{\partial \xi\partial \delta}
=
-\sigma\xi^{-1}
\sum_{i=1}^{n}
x_i|x_i|^{\delta} \ln|x_i|
\big(\sigma x_i|x_i|^{\delta}-\mu\big)
\Psi^{-1}_i(\Theta)
[\xi \Omega_i(\Theta)-\Psi_i^{-{1}/{\xi}}(\Theta)].
\end{array}
\end{align}

If $X\sim F_{\rm BGEV}$, under certain regularity conditions, the Fisher information matrix (FIM) may also be written as
\begin{align*}
[\mathcal{I}(\Theta)]_{i,j}
=
-\mathbb{E}
\biggl[\frac{\partial^2 \ell_1(\Theta;X)}{\partial \theta_i\partial \theta_j}\,\bigg\vert\, \Theta\biggr],
\quad \theta_i,\theta_j\in \{\xi,\mu,\sigma, \delta\},
\end{align*}
where $\partial^2 \ell_1(\Theta;x)/\partial \theta_i\partial \theta_j$ are given in \eqref{sec-der-1} and \eqref{sec-der-2} by taking $n=1$.
It is clear that the analytical calculations of the FIM above are difficult. But by using the Monte Carlo estimates of the Hessian of the negative log-likelihood function we can form an average of these ones as an estimate of the FIM (Spall 2005).

\section{Numerical Illustrations}
\label{sec:6}
\noindent

The performance of maximum likelihood estimators calculated was tested by Monte Carlo simulation with  45 combinations of  $\xi, \mu$ and $  \delta$. We consider the scale parameter always equal to 1; $ \sigma = 1 $.

We use the following procedure:
\\
(i) Generate random samples of  $X\sim F_{\rm BGEV}(\cdot; \xi,\mu,\sigma, \delta)$ by using the quantile function \eqref{Q}.
 Were generated M=100 random samples (replications) of size N=50, 100,  250, 1000, for each choice of  $\theta$.
\\
(ii) Calculate the maximum likelihood estimates (MLEs) of $ \theta= ( \xi,\mu, \delta)$ using  the  equations \eqref{fv1}- \eqref{xi}. Here, use  the Nelder and Mead' optimization method  implemented in the R statistical program (Nelder and Mead 1965). 
\\
(iii) For the initial value, consider the real parameter value  added by a value of a auniform variable in $(0,1)$.

The estimation results are found in Tables \ref{Tablexi1} - \ref {Tablexim05}. We have five different configurations. In configuration 1,
 the shape parameters  $\xi=1$, $\mu \in \{ 1,0,-1\}$ and $\delta \in\{ 0,2,4\}$. Its mean estimates,  the bias and the mean square error (MSE) of the estimates are present in Table \ref{Tablexi1}. In configuration 2,   $\xi=0.5$, $\mu \in \{ 1,0,-1\}$ and $\delta \in\{ 0,2,4\}$. For this configuration, the  mean estimates,  the bias and the MSE of the estimates are present in Table \ref{Tablexi05}. In cofigurations 3, 4, and 5, $\mu \in \{ 1,0,-1\}$, $\delta \in\{ 0,2,4\}$ and $\xi=0.25, \ -0.25$, and $-0.5$, respectively. The  mean estimates,  the bias and the MSE of the estimates for  these three configurations are shown in the Tables \ref{Tablexi025}, \ref{Tablexim025}, and \ref{Tablexim05}.
 In the five configurations, the bias and MSE of the mean estimates are small. That is, the algorithm has obtained good estimates.
 The bias and the MSE decreased as the sample size was increased from 50 to 1000.

\begin{table}[!htbp]
	\centering
	\caption{Empirical means, bias and mean squared errors (MSE) of the
		estimates of the parameters for $(\xi, \sigma) = (1, 1)$ and some values of $\mu, \delta$ and $n$.}
	\resizebox{\linewidth}{!}{
		\begin{tabular}{cccrrcccrrrrrcccc}
			\hline
			&\multicolumn{6}{c}{Empirical means} && \multicolumn{3}{c}{Bias} && \multicolumn{3}{c}{MSE} \\
			$n$&$\xi$ &  $\widehat{\xi}$  & $\mu$ & $\widehat{\mu}$ &  $\delta$ & $\widehat{\delta}$  && $\widehat{\xi}$ & $\widehat{\mu}$ &  $\delta$  && $\widehat{\xi}$ & $\widehat{\mu}$ &  $\widehat{\delta}$  \\ \hline
			&1 & 1.018& $-$1 &$-$1.017 & 0 & 0.036 && $-$0.018& 0.017&$-$0.036&& 0.049& 0.0214& 0.0171\\
			&1 & 1.038& &$-$1.013& 2 &  2.082 && $-$0.038& 0.013& $-$0.082&& 0.039& 0.0206& 0.1167\\
			&1 &1.034 &$-$1& $-$1.053& 4 & 4.292	&& $-$0.034& 0.053& $-$0.292&& 0.0449& 0.0201 &0.3219\\
			50		&1 &1.085 & 0 &0.007 & 0 &0.048	&&$-$0.085 &$-$0.007 &$-$0.048 &&0.0527 &0.0081& 0.0229\\
			&1 &1.085 &0 &0.009 &2 &2.126 &&$-$0.085 &$-$0.009 &$-$0.126 &&0.0433 &0.007& 0.1972 \\
			&1 &1.049 & 0 &0.028 & 4 &4.152 	&&$-$0.049 &$-$0.028 &$-$0.152 &&0.0546 &0.0078& 0.4478\\
			&1 &1.148 &1 &0.97  &0 &0.096 &&$-$0.148 &0.03 &$-$0.096 &&0.0797 &0.029 &   0.0668\\
			&1 &1.083 &1 & 0.995 & 2 &2.14 &&$-$0.083 & 0.005 &$-$0.14 && 0.0867 & 0.0279& 0.4675\\
			&1 &1.125 &1 &0.973 & 4 &4.452	 &&$-$0.125 & 0.027 &$-$0.452&& 0.0954& 0.0295& 2.3334\\ \hline
			&1 & 1.018& $-$1 &$-$1.002& 0 &0.009&& $-$0.018& 0.002& $-$0.009&& 0.0195& 0.0073& 0.0069\\
			&1 & 1.037& $-$1 &$-$1.007& 2 &2.025&& $-$0.037& 0.007& $-$0.025&& 0.0151& 0.0103& 0.0635\\
			&1 & 1.022&$-$1 &$-$1.007 &4 &4.114 &&$-$0.022 &0.007 &$-$0.114 &&0.019 &0.0092& 0.1591\\
			100		&1 & 1.038& 0 &0.005  &0  &0.024&& $-$0.038& $-$0.005& $-$0.024&& 0.0189& 0.0027& 0.0086\\
			&1 & 1.07 & 0 &$-$0.005 &2 &2.079 &&$-$0.07 &0.005 &$-$0.079 &&0.021& 0.0034& 0.089\\
			&1 & 1.029& 0 &0.005  &4  &4.079 && $-$0.029& $-$0.005& $-$0.079&& 0.0177 &0.0039& 0.2877\\
			&1 & 1.08 & 1 &0.97   &0  &0.071 &&$-$0.08& 0.03& $-$0.071&& 0.034& 0.0142& 0.0295\\
			&1 & 1.076& 1 &0.978  &2  &2.199 &&$-$0.076& 0.022& $-$0.199&& 0.037& 0.0124& 0.2379\\
			&1 & 1.036& 1 &0.995  &4  &4.199 &&0.036& 0.052& $-$0.159&& 0.037& 0.0124& 0.0237\\ \hline
			&1&  0.999& $-$1&$-$1.003 & 0& 0.003 &&  0.013& 0.001 &0.003  &&$-$0.013 &0.0048 &0.0038\\
			&1&  1.006& $-$1& $-$1.002& 2&  2.026&& $-$0.006& 0.002 &$-$0.026 &&0.0046 &0.0029 &0.0243\\
			&1&  1.018& $-$1& $-$1.01 & 4& 4.051 &&$-$0.018 &0.01   &$-$0.051 &&0.0069 &0.0036 &0.0744\\
			250		&1&  1.028& 0 & $-$0.001& 0& 0.019 &&$-$0.028 &0.001  &$-$0.019 &&0.0065 &0.0011 &0.0042\\
			&1&  1.031& 0 & 0.001 & 2& 2.043 &&$-$0.031 &$-$0.001 &$-$0.043 &&0.0078 &0.0012 &0.0318\\
			&1&  1.016& 0 & 0.001 & 4& 4.065 &&$-$0.016 &$-$0.001 &$-$0.065 &&0.0068 &0.0014 &0.1125\\
			&1&  1.044& 1 & 0.983 & 4& 0.039 &&$-$0.044 &0.017  &$-$0.039 &&0.0181 &0.0079 &0.0165\\
			&1&  1.009& 1 & 1.005 & 2& 2.045 &&$-$0.009 &$-$0.005 &$-$0.045 &&0.0138 &0.0067 &0.1171\\
			&1&  1.031& 1 & 0.991 & 4& 4.059 &&$-$0.031 &0.009  &$-$0.059 &&0.0125 &0.0057 &0.2744\\ \hline
			&1    &  1.005  &   1  &   0.996 &  0 & 0.012 &&  $-$0.005  &    0.004  &   $-$0.012&&   0.0028 &   0.0014  &   0.0025      \\
			&1    &  1.011  &   1  &   0.995 &  2 & 2.024 &&  $-$0.011  &    0.005  &  $-$0.024&&  0.0025 &   0.0014  &   0.0213     \\
			&1    &  1.002  &   1  &   1.001 &  4 & 3.998 && $-$0.002  &    $-$0.001  &    0.002&&  0.0029 &   0.0014  &   0.0578      \\
			1000	&1    &  1.01   &   0  &   0     &  0 & 0.005 &&  $-$0.010  &    0.000  &   $-$0.005&&   0.002  &    4e$-$04   &  0.0010      \\
			&1    &  1.011  &   0  &  $-$0.003 &  2 & 2.00  &&  $-$0.011  &     0.003  &   $-$0.001&&  0.0019 &    4e$-$04   &  0.0064       \\
			&1    &  1.008  &   0  &   0.001 &  4 & 4.039 && $-$0.008  &    $-$0.001  &  $-$0.039&&  0.0015 &    4e$-$04   &   0.0216       \\
			&1    &  1.001  &  $-$1  &  $-$1.006 &  0 & 0.006 &&   $-$0.001  &    0.006  &   $-$0.006&&  0.0019 &    9e$-$04   &   6e$-$040     \\
			&1    &  0.998  &  $-$1  &  $-$1.002 &  2 & 2.019 &&   0.002  &    0.002  &  $-$0.019&& 0.0011 &    9e$-$04   &  0.0059        \\
			&1    &  1.001  &  $-$1  &  $-$1.001 &  4 & 4.012 && $-$0.001  &    0.001  &   $-$0.012&& 0.0014 &   6e$-$04   &   0.0140    \\	\hline
		\end{tabular}\label{Tablexi1}
	}
\end{table}

\begin{table}[!htbp]
	\centering
	\caption{Empirical means, bias and mean squared errors (MSE) of the
		estimates of the parameters for $(\xi, \sigma) = (0.5, 1)$ and some values of $\mu, \delta$ and $n$.}
	\resizebox{\linewidth}{!}{
		\begin{tabular}{cccrrcccrrrrrcccc}
			\hline
			&\multicolumn{6}{c}{Empirical means} && \multicolumn{3}{c}{Bias} && \multicolumn{3}{c}{MSE} \\
			
			$n$&$\xi$ &  $\widehat{\xi}$  & $\mu$ & $\widehat{\mu}$ &  $\delta$ & $\widehat{\delta}$  && $\widehat{\xi}$ & $\widehat{\mu}$ &  $\delta$  && $\widehat{\xi}$ & $\widehat{\mu}$ &  $\widehat{\delta}$  \\ 		\hline
			&0.5&  0.501 &$-$1& $-$0.999& 0& 0.032&& $-$0.001& $-$0.001& $-$0.032&& 0.0233& 0.024& 0.0129\\
			&0.5&  0.524 &$-$1& $-$1.032& 2& 2.15 &&$-$0.024 &0.032 &$-$0.15 &&0.0268& 0.025& 0.1033\\
			&0.5&  0.486 &$-$1& $-$1.019& 4& 4.215&& 0.014 &0.019 &$-$0.215 &&0.0254& 0.023& 0.2411\\
			50	&0.5&  0.549 &0 & 0.064 & 0& $-$0.049&& 0.028& $-$0.064& 0.0375&& 0.0154& 0.0282 &0321\\
			&0.5&  0.562 &0 & 0.015 & 2& 2.118 &&$-$0.062& $-$0.015& $-$0.118&& 0.045& 0.0175& 0.1552\\
			&0.5&  0.551 &0 & 0.019 & 4& 4.282 &&$-$0.051& $-$0.019& $-$0.282&& 0.0303& 0.0124& 0.4835\\
			&0.5&  0.541 &1 & 1.029 & 0& 0.057 &&$-$0.041& $-$0.029& $-$0.057&& 0.0725& 0. 0255& 0.0515\\
			&0.5&  0.559 &1 & 1.006 & 2& 2.094 &&$-$0.059& $-$0.006& $-$0.094&& 0.0786& 0.0254& 0.3663\\
			&0.5&  0.563 &1 & 0.998 & 4& 4.379 &&$-$0.063& 0.002 &$-$0.379 &&0.083& 0.0238 &1.7225\\ \hline
			&0.5&  0.503& $-$1& $-$1.011& 0& 0.024&& $-$0.003& 0.011& $-$0.024 &&0.0113& 0.0142& 0.0068\\
			&0.5&  0.503& $-$1& $-$1.02 & 2& 2.09 && $-$0.003& 0.02 &$-$0.09 &&0.0088 &0.0128& 0.0469\\
			&0.5&  0.518& $-$1& $-$1.017& 4& 4.145&& $-$0.018& 0.017& $-$0.145&& 0.0094& 0.011& 0.1269\\
			100	&0.5& 0.525 & 0 &  0.005& 0& 0.025&& $-$0.025& $-$0.005& $-$0.025&& 0.0115& 0.0052& 0.0084\\
			&0.5& 0.511 & 0 &  0.021& 2& 2.067&& $-$0.011& $-$0.021& $-$0.067&& 0.0105& 0.0065& 0.0663\\
			&0.5& 0.52  & 0 &  0.012& 4& 4.136&& $-$0.02 &$-$0.012 &$-$0.136 &&0.0131 &0.0057& 0.2212\\
			&0.5& 0.53  & 1 & 1.012 & 0& 0.011&& $-$0.03 &$-$0.012 &$-$0.011 &&0.0281 &0.01 &0.0195\\
			&0.5& 0.53  & 1 & 0.997 & 2& 2.102&& $-$0.03 &0.003 &$-$0.102 &&0.0519 &0.0125 &0.3862\\
			&0.5& 0.516 & 1 & 1.002 & 4& 4.093&& $-$0.016& $-$0.002& $-$0.093&& 0.024 &0.0084& 0.3418\\ 		\hline
			&0.5&  0.513&  $-$1& $-$1.002&0& 0.011&& 0.002 &$-$0.011& 0.0042&& 0.0033& 0.002& 0.232\\
			&0.5&  0.505& $-$1 &$-$1.009 &2& 2.027&& $-$0.005& 0.009& $-$0.027&& 0.0051& 0.004 &0.0164\\
			&0.5&  0.501& $-$1 &$-$1.009 &4& 4.044&& $-$0.001& 0.009& $-$0.044&& 0.0042& 0.0054& 0.0571\\
			250		&0.5&  0.511&  0 & 0.009 &0& 0.01 &&$-$0.011 &$-$0.009& $-$0.01 &&0.0033 &0.0024 & 0.0035\\
			&0.5&  0.505&  0 & 0.008 &2& 2.035&& $-$0.005& $-$0.008& $-$0.035&& 0.0053& 0.0026& 0.0315\\
			&0.5&  0.514&  0 &$-$0.005 &4& 4.051&& $-$0.014& 0.005 &$-$0.051&& 0.006 &0.003 &0.0654\\
			&0.5&  0.537&  1 & 0.988 &0& 0.023&& $-$0.037& 0.012 &$-$0.023&& 0.0108 &0.0045& 0.0066\\
			&0.5&  0.511&  1 & 1.005 &2& 2.008&& $-$0.011& $-$0.005& $-$0.008&& 0.0109& 0.0047& 0.0639\\
			&0.5&  0.528&  1 & 0.998 &4& 4.062&& $-$0.028& 0.002 &$-$0.062&& 0.0108 &0.0046 &0.1906\\ \hline
			&0.5  &  0.493  &   1  &   1.003 &  0 &$-$0.003&& 0.007  &   $-$0.003  &    0.003&&0.0026 &    0.0015  &   0.0020 \\
			&0.5  &  0.512  &   1  &   0.992 &  2 & 2.041&&  $-$0.012  &    0.008  &  $-$0.041&& 0.0023 &    0.0013  &   0.0167 \\
			&0.5  &  0.506  &   1  &   0.999 &  4 & 4.023&&$-$0.006  &    0.001  &   $-$0.023 && 0.0025 &    0.0011  &    0.0497\\
			1000	&0.5  &  0.507  &   0  &  $-$0.005 &  0 &$-$0.001&& $-$0.007  &    0.005  &    0.001 &&8e$-$04  &    6e$-$04   &    9e$-$04 \\
			&0.5  &  0.501  &   0  &   0.003 &  2 & 2.006&& $-$0.001  &    $-$0.003 &  $-$0.006 && 0.0013 &    7e$-$04   &    0.0061 \\
			&0.5  &  0.505  &   0  &  $-$0.001 &  4 & 3.998&&  $-$0.005  &    0.001  &   0.002 &&9e$-$04  &    7e$-$04   &    0.0135 \\
			&0.5  &  0.501  &  $-$1  &  $-$0.998 &  0 & 0    && $-$0.001  &   $-$0.002  &    0.000 && 0.0011 &    0.0012  &   6e$-$04  \\
			&0.5  &  0.505  &  $-$1  &  $-$1.001 &  2 & 2.004 && $-$0.005  &   0.001   &   $-$0.004 &&  9e$-$04  &    9e$-$04   &   0.0045 \\
			&0.5  &  0.499  &  $-$1  &  $-$0.997 &  4 & 4.007 && 0.001   &  $-$0.003   &   $-$0.007 &&  0.001  &   0.0013  &    0.0149\\
			\hline
		\end{tabular}\label{Tablexi05}
	}
\end{table}


\begin{table}[!htbp]
	\centering
	\caption{Empirical means, bias and mean squared errors (MSE) of the
		estimates of the parameters for $(\xi, \sigma) = (0.25, 1)$ and some values of $\mu, \delta$ and $n$.}
	\resizebox{\linewidth}{!}{
		\begin{tabular}{cccrrcccrrrrrcccc}
			\hline
			&\multicolumn{6}{c}{Empirical means} && \multicolumn{3}{c}{Bias} && \multicolumn{3}{c}{MSE} \\
			
			$n$&$\xi$ &  $\widehat{\xi}$  & $\mu$ & $\widehat{\mu}$ &  $\delta$ & $\widehat{\delta}$  && $\widehat{\xi}$ & $\widehat{\mu}$ &  $\delta$  && $\widehat{\xi}$ & $\widehat{\mu}$ &  $\widehat{\delta}$  \\		\hline
			&0.25&  0.256& $-$1& $-$1.008& 0 & 0.04 && $-$0.045& 0.128& $-$0.04&& 0.0216& 0.0322& 0.0116\\
			&0.25&  0.244& $-$1& $-$1.018& 2 & 2.116&& 0.006 & 0.018& $-$0.116&& 0.0162& 0.0323& 0.0838\\
			&0.25&  0.239& $-$1& $-$1.013& 4 & 4.185&& 0.011 & 0.013& $-$0.185 &&0.0186& 0.0262& 0.2052\\
			50		&0.25&  0.273&  0&  0.016& 0 & 0.056&& $-$0.023& $-$0.016& $-$0.056&& 0.0302& 0.0155& 0.0251\\
			&0.25&  0.285&  0& 0.026 & 2 & 2.131&& $-$0.035& $-$0.026& $-$0.131&& 0.035 &0.0168 &0.2199\\
			&0.25&  0.275&  0& 0.02  & 4 & 4.213&& $-$0.025& $-$0.02 &$-$0.213 &&0.0247 &0.0143 &0.4635\\
			&0.25&  0.267&  1& 1.001 & 0 & 0.037&& $-$0.017& $-$0.001& $-$0.037&& 0.0803& 0.029 &0.0414\\
			&0.25&  0.289&  1& 1.017 & 2 &2.197 &&$-$0.039 &$-$0.017 &$-$0.197 &&0.0827 &0.0283 &0.5167\\
			&0.25&  0.292&  1& 1.009 & 4 &4.309 &&$-$0.042 &$-$0.009 &$-$0.309 &&0.0812 &0.0226 &1.3124	\\ \hline
			&0.25& 0.245& $-$1 &$-$1.021& 0& 0.037&& 0.005& 0.021& $-$0.037 &&0.0096& 0.0199& 0.0054\\
			&0.25& 0.247& $-$1 &$-$1.007& 2& 2.071&& 0.003& 0.007& $-$0.071 &&0.0075& 0.0157& 0.048\\
			&0.25& 0.241& $-$1 &$-$1.004& 4& 4.045&& 0.009& 0.004& $-$0.045 &&0.007 &0.0114 &0.1047\\
			100		&0.25& 0.255&  0 & 0.027& 0& 0.028&& $-$0.005& $-$0.027& $-$0.028&& 0.0104& 0.0074& 0.007\\
			&0.25& 0.262&  0 & 0.01 & 2& 2.067&& $-$0.012& $-$0.01& $-$0.067 &&0.0107 &0.0087& 0.0579\\
			&0.25& 0.263&  0 & 0.007& 4& 4.115&& $-$0.013& $-$0.007& $-$0.115&& 0.0101& 0.0085& 0.1793\\
			&0.25& 0.249&  1 & 1.009& 0& 0.011&& 0.001 &$-$0.009 &$-$0.011 &&0.025 &0.0108 &0.0124\\
			&0.25& 0.272&  1 & 1.098& 2& 2.058&& $-$0.022&  $-$0.058& 0.0325&& 0.0117& 0.1798& 0.234\\
			&0.25& 0.241&  1 & 1.021& 4& 4.052&& 0.009 &$-$0.021& $-$0.052&& 0.0232& 0.0104& 0.404\\ \hline
			&0.25&  0.2255& $-$1& $-$0.994& 0& 0.007&& $-$0.005& $-$0.006& $-$0.007 &&0.0022 &0.0043& 0.0017\\
			&0.25&  0.243 & $-$1& $-$1.003& 2& 2.053&& 0.007 &0.003 &$-$0.053&& 0.0026& 0.0046& 0.0186\\
			&0.25&  0.254 & $-$1& $-$1.014& 4& 4.072&& $-$0.004& 0.014 &$-$0.072&& 0.003& 0.0052& 0.0436\\
			250		&0.25&  0.267 & 0 & 0.002 & 0& 0.02 &&$-$0.017 &$-$0.002 &$-$0.02 &&0.0028& 0.0033& 0.0032\\
			&0.25&  0.255 & 0 & 0.005 & 2& 2.01 &&$-$0.005 &$-$0.005 &$-$0.01 &&0.0034& 0.0026& 0.0229\\
			&0.25&  0.249 & 0 & 0.001 & 4& 4.038 &&0.001 &$-$0.001 &$-$0.038&& 0.0046& 0.0031& 0.0772\\
			&0.25&  0.236 & 1 & 1.01  & 0& $-$0.003&& 0.014& $-$0.01 &0.003 &&0.0079 &0.005 &0.0038\\
			&0.25&  0.241 & 1 & 1.001 &2 & 1.997 &&0.009 &$-$0.001 &0.003 &&0.0068 &0.0037& 0.0469\\
			&0.25&  0.269 & 1 &0.998  &4 & 4.121 &&$-$0.019& 0.002 &$-$0.121 &&0.0087& 0.0049& 0.195\\ 		\hline
			&0.25 &  0.253  &   1  &   1.002 &  0 & 0.002&&$-$0.003  &  $-$0.002  &  $-$0.002&& 0.0015 &    0.001   &   0.0010 \\
			&0.25 &  0.25   &   1  &   0.996 &  2 & 2 &&   0.000  &    0.004  &  0.000  &&  0.002  &    9e$-$04   &   0.0132   \\
			&0.25 &  0.256  &   1  &   0.996 &  4 & 4.018 && $-$0.006  &    0.004  & $-$0.018		&&  0.0018 &    0.0013  &   0.0331 \\
			1000	&0.25 &  0.254  &   0  &  $-$0.001 &  0 & 0.002&& $-$0.004  &    0.001  &  $-$0.002 && 9e$-$04  &   8e$-$04   &    6e$-$04 \\
			&0.25 &  0.249  &   0  &   0.004 &  2 & 2.013 && 0.001  &   $-$0.004  &  $-$0.013 && 8e$-$04  &   8e$-$04   &   0.0066\\
			&0.25 &  0.252  &   0  &  $-$0.002 &  4 & 4.012&& $-$0.002  &    0.002  & $-$0.012  && 8e$-$04  &   6e$-$04   &   0.0199\\
			&0.25 &  0.248  &  $-$1  &  $-$1.001 &  0 & 0.002 && 0.002  &     0.001  & $-$0.002 && 6e$-$04  &   0.0014  &   5e$-$04\\
			&0.25 &  0.25   &  $-$1  &  $-$0.998 &  2 & 2.004 && 0.000  &    $-$0.002  & $-$0.004 && 7e$-$04  &   0.0011  &   0.0037\\
			&0.25 &  0.253  &  $-$1  &  $-$1.002 &  4 & 4.015 && $-$0.003  &    0.002  &  $-$0.015 &&  7e$-$04  &   0.001   &   0.0106\\
			\hline
		\end{tabular}\label{Tablexi025}
	}
\end{table}


\begin{table}[!htbp]
	\centering
	\caption{Empirical means, bias and mean squared errors (MSE) of the
		estimates of the parameters for $(\xi, \sigma) = (-0.25, 1)$ and some values of $\mu, \delta$ and $n$.}
	\resizebox{\linewidth}{!}{
		\begin{tabular}{cccrrcccrrrrrcccc}
			\hline
			&\multicolumn{6}{c}{Empirical means} && \multicolumn{3}{c}{Bias} && \multicolumn{3}{c}{MSE} \\
			
			$n$		&		$\xi$ &  $\widehat{\xi}$  & $\mu$ & $\widehat{\mu}$ &  $\delta$ & $\widehat{\delta}$  && $\widehat{\xi}$ & $\widehat{\mu}$ &  $\delta$  && $\widehat{\xi}$ & $\widehat{\mu}$ &  $\widehat{\delta}$  \\
			\hline
			&$-$0.25&  $-$0.266& $-$1& $-$0.981& 0& 0.055&& 0.016& $-$0.019& $-$0.055&& 0.0071& 0.0269& 0.0197\\
			&$-$0.25&  $-$0.288& $-$1& $-$1.012& 2& 2.151&& 0.038& 0.012& $-$0.151&& 0.0082& 0.0279& 0.0849\\
			&$-$0.25&  $-$0.296& $-$1& $-$0.997& 4& 4.288&& 0.046& $-$0.003& $-$0.288&& 0.0093& 0.0259& 0.2736\\
			50		&$-$0.25&  $-$0.286&  0& 0.002&  0& 0.024&& 0.036& $-$0.002& $-$0.024&& 0.0159& 0.0298& 0.018\\
			&$-$0.25&  $-$0.266&  0& 0.003&  2& 2.094&& 0.016& $-$0.003& $-$0.094&& 0.0125& 0.0243& 0.1833\\
			&$-$0.25 & $-$0.283&  0& 0.031&  4& 4.079 &&0.033& $-$0.031& $-$0.079&& 0.0163& 0.0252& 0.439\\
			&$-$0.25 & $-$0.256		&1& 1.01 &  0&  0.028&& 0.006& $-$0.01& $-$0.028&& 0.0326& 0.0312& 0.0211\\
			&$-$0.25 & $-$0.278&  1& 1.03 &  2& 2.119 &&0.028& $-$0.03 &$-$0.119&& 0.0332& 0.0224& 0.2411\\
			&$-$0.25 & $-$0.257&  1& 1.037&  4& 4.165 &&0.007& $-$0.037 &$-$0.165 &&0.0374& 0.0359& 0.6435		\\ \hline
			&$-$0.25& $-$0.264& $-$1& $-$1.006& 0& 0.027&& 0.014& 0.006& $-$0.027&& 0.0033& 0.0149& 0.0051\\
			&$-$0.25& $-$0.264& $-$1& $-$0.992& 2& 2.056&& 0.014& $-$0.008& $-$0.056&& 0.0032& 0.0169& 0.0435\\
			&$-$0.25& $-$0.272& $-$1& $-$1.002& 4& 4.156&& 0.022& 0.002& $-$0.156&& 0.0041& 0.0116& 0.134\\
			100		&$-$0.25& $-$0.258&  0& 0.002 & 0& 0.028&& 0.008& $-$0.002& $-$0.028&& 0.0053& 0.0138& 0.0095\\
			&$-$0.25& $-$0.271&  0& 0.015 & 2& 2.02 &&0.021 & $-$0.015 &$-$0.02&& 0.0056& 0.014& 0.0709\\
			&$-$0.25& $-$0.255&  0& $-$0.001& 4& 4.095&& 0.005& 0.001 &$-$0.095&& 0.005& 0.0118& 0.2409\\
			&$-$0.25& $-$0.256&  1& 1.018 & 0& 0.004&& 0.006& $-$0.018& $-$0.004&& 0.0115& 0.0118& 0.0093\\
			&$-$0.25& $-$0.275&  1& 1.026 &2 &2.007 &&0.025 & $-$0.026& $-$0.007&& 0.0082& 0.013 &0.0599\\
			&$-$0.25& $-$0.274&  1& 1.015 & 4& 4.02 &&0.024 &$-$0.015 &$-$0.02  &&0.0103& 0.0131& 0.234		\\ \hline
			&$-$0.25& $-$0.261&  $-$1& $-$0.994& 0& 0.012&& 0.011& $-$0.006& $-$0.012 &&9e$-$04& 0.0041& 0.0014\\
			&$-$0.25& $-$0.255&  $-$1& $-$0.985& 2& 2.03&& 0.005 &$-$0.015 &$-$0.03  &&0.0011& 0.0043& 0.0166\\
			&$-$0.25& $-$0.259&  $-$1& $-$1.007& 4& 4.037&& 0.009& 0.007 &$-$0.037 &&0.0012& 0.0058& 0.0458\\
			250		&$-$0.25& $-$0.257&   0& 0.012 & 0& 0.007&& 0.007&$-$0.012 &0.0018 &&0.0045& 0.0108&0.0769\\
			&$-$0.25& $-$0.257&   0& 0.014 & 2& 2.03 &&0.007 &$-$0.014 &$-$0.03 &&0.0021& 0.0052& 0.027\\
			&$-$0.25& $-$0.258&   0& 0.014 & 4& 4.023&& 0.008& $-$0.014 &$-$0.023&& 0.0022& 0.0067& 0.1011\\
			&$-$0.25& $-$0.258&   1& 1.012 & 0& 0.006&& 0.008& $-$0.012 &$-$0.006&& 0.0044& 0.0046& 0.0035\\
			&$-$0.25& $-$0.268&   1& 0.997 & 2& 1.974&& 0.018& 0.003& 0.026&& 0.0031& 0.0038& 0.026\\
			&$-$0.25& $-$0.256&   1& 1.007 & 4& 4.003&& 0.006& $-$0.007& $-$0.003&& 0.0035& 0.0345& 0.0042\\ \hline
			&$-$0.25 & $-$0.257  &   1  &   1.003 &  0 &$-$0.002&& 0.007  &   $-$0.003  &   0.002&& 9e$-$04  &    0.0013  &   9e$-$04\\
			&$-$0.25 & $-$0.248  &   1  &   1.003 &  2 & 2.01&& $-$0.002  &   $-$0.003  &   $-$0.010&& 8e$-$04  &   0.0013  &   0.0061 \\
			&$-$0.25 & $-$0.253  &   1  &   1     &  4 & 4.003&& 0.003  &    0.000  &   $-$0.003&&  0.001  &    0.0012  &  0.0228 \\
			1000	&$-$0.25 & $-$0.251  &   0  &   0     &  0 & 0.002&& 0.001  &     0.000  &  $-$0.002&& 4e$-$04  &    0.0013  &   7e$-$04 \\
			&$-$0.25 & $-$0.252  &   0  &   0     &  2 & 2.004&& 0.002  &     0.000  &  $-$0.004&& 5e$-$04  &   0.0012  &    0.0087 \\
			&$-$0.25 & $-$0.251  &   0  &   0.001 &  4 & 4.003&& 0.001  &   $-$0.001  &   $-$0.003&& 4e$-$04  &     0.0011  &  0.0176 \\
			&$-$0.25 & $-$0.25   &  $-$1  &  $-$1.003 &  0 & 0.003&& 0.000  &   0.003  &   $-$0.003&&  2e$-$04  &    0.003   &    0.0013 \\
			&$-$0.25 & $-$0.253  &  $-$1  &  $-$1     &  2 & 1.998&& 0.003  &   0.000  &    0.002&& 3e$-$04  &   0.0016  &   0.0033\\
			&$-$0.25 & $-$0.254  &  $-$1  &  $-$0.997 &  4 & 4.021&& 0.004  &   $-$0.003  &  $-$0.021 && 2e$-$04  &   0.0012  &   0.0088\\
			\hline
		\end{tabular}\label{Tablexim025}
	}
\end{table}


\begin{table}[!htbp]
	\centering
	\caption{Empirical means, bias and mean squared errors (MSE) of the
		estimates of the parameters for $(\xi, \sigma) = (-0.5, 1)$ and some values of $\mu, \delta$ and $n$.}
	\resizebox{\linewidth}{!}{
		\begin{tabular}{cccrrcccrrrrrcccc}
			\hline
			&\multicolumn{6}{c}{Empirical means} && \multicolumn{3}{c}{Bias} && \multicolumn{3}{c}{MSE} \\
			
			$n$		&$\xi$ &  $\widehat{\xi}$  & $\mu$ & $\widehat{\mu}$ &  $\delta$ & $\widehat{\delta}$  && $\widehat{\xi}$ & $\widehat{\mu}$ &  $\widehat{\delta}$  && $\widehat{\xi}$ & $\widehat{\mu}$ &  $\widehat{\delta}$  \\ \hline
			&$-$0.5& $-$0.553& $-$1& $-$0.955& 0& 0.044&& 0.053& $-$0.045& $-$0.044&& 0.0102& 0.0305& 0.0123\\
			&$-$0.5& $-$0.563& $-$1& $-$0.948& 2& 2.229&& 0.063& $-$0.052& $-$0.229&& 0.0129& 0.0319& 0.1927\\
			&$-$0.5& $-$0.579& $-$1& $-$0.952& 4& 4.25 &&0.079 &$-$0.048 &$-$0.25&& 0.0077& 0.0209& 0.3505\\
			$50$	&$-$0.5& $-$0.545&  0&  0.04 & 0& 0.017&& 0.045 &$-$0.04& $-$0.017&& 0.012& 0.0259& 0.0197\\
			&$-$0.5& $-$0.565&  0&  0.032& 2& 2.005&& 0.065 &$-$0.032& $-$0.005&& 0.0112& 0.0224& 0.1354\\
			&$-$0.5& $-$0.553& $-$1&$-$0.553&  0&  0.0434&& 4.069& 0.053& $-$0.043&& $-$0.069& 0.0127& 0.0275\\
			&$-$0.5& $-$0.562&  1&  1.021& 0& $-$0.012&& 0.062 & $-$0.021& 0.012&& 0.0258& 0.0248& 0.0211\\
			&$-$0.5& $-$0.532&  1&  1.025& 2& 2.067 &&0.032& $-$0.025& $-$0.067&& 0.0247& 0.0286& 0.1658\\ 	\hline	
			&$-$0.5&  $-$0.536& $-$1& $-$0.964& 0& 0.025&& 0.036& $-$0.036& $-$0.025&& 0.0036& 0.0131& 0.0052\\
			&$-$0.5&  $-$0.528& $-$1& $-$0.982& 2& 2.076&& 0.028& $-$0.018& $-$0.076&& 0.003& 0.0114& 0.0457\\
			&$-$0.5&  $-$0.531& $-$1& $-$0.984& 4& 4.182&& 0.031& $-$0.016& $-$0.182&& 0.0027& 0.0093& 0.1574\\
			$100$	&$-$0.5&  $-$0.523&  0&  0.015& 0& 0.015&& 0.023& $-$0.015& $-$0.015&& 0.004& 0.0113& 0.008\\
			&$-$0.5&  $-$0.532&  0&  0.026& 2& 2.053&& 0.032& $-$0.026& $-$0.053&& 0.0042& 0.011& 0.0643\\
			&$-$0.5&  $-$0.513&  0&  0.008& 4& 4.153&& 0.013& $-$0.008& $-$0.153&& 0.0042& 0.0101& 0.2634\\
			&$-$0.5& $-$0.516 &  1&  1.01 & 0& 0.016&& 0.016 &$-$0.01 &$-$0.016&& 0.011 &0.0113& 0.0102\\
			&$-$0.5& $-$0.532 &  1&  1.008& 2& 1.988&& 0.032 &$-$0.008& 0.012 &&0.0095 &0.0116& 0.0731\\
			&$-$0.5& $-$0.527 &  1&  0.983& 4& 3.955&& 0.027 &0.017 &0.045 &&0.0089 &0.0103& 0.1963		\\ \hline
			&$-$0.5 &$-$0.516& $-$1& $-$0.985&  0& 0.018&& 0.016& $-$0.015& $-$0.018&& 7e$-$04& 0.0033& 0.002\\
			&$-$0.5 &$-$0.516& $-$1& $-$0.982&  2& 2.049&& 0.016& $-$0.018& $-$0.049&& 0.0012& 0.0096& 0.0346\\
			&$-$0.5 &$-$0.51 &$-$1 & $-$1.004&  4& 4.096&& 0.01& 0.004& $-$0.096&& 8e$-$04& 0.0035& 0.0767\\
			$250$	&$-$0.5 &$-$0.514& 0 & 0.011 &  0& $-$0.003&& 0.014& $-$0.011& 0.003&& 0.0012& 0.0033& 0.003\\
			&$-$0.5 &$-$0.51 & 0 & 0.007 &  2& 2.025 &&0.01& $-$0.007& $-$0.025&& 0.0015& 0.0042& 0.0352\\
			&$-$0.5 &$-$0.509& 0 &$-$0.001 &  4& 4.009 && 0.009& 0.001& $-$0.009&& 0.0012& 0.0038& 0.0823\\
			&$-$0.5 &$-$0.51 & 1 &1.011  &  0& 0.006 &&0.01& $-$0.011& $-$0.006&& 0.0031& 0.0068& 0.0031\\
			&$-$0.5 &$-$0.513& 1 &0.999  &  2& 1.99 &&0.013& 0.001& 0.01&& 0.0028& 0.004& 0.0283\\
			&$-$0.5 &$-$0.511& 1 &1.007  &  4& 3.995 &&0.011& $-$0.007& 0.005&& 0.0033& 0.0061& 0.0712\\ \hline
			&$-$0.5  & $-$0.505  &   1  &   1.006 &  0 & 0&& 0.005  &   $-$0.006  &    0.000 &&  6e$-$04  &   0.0012  &   8e$-$04\\
			&$-$0.5  & $-$0.502  &   1  &   1.003 &  2 & 2.002&& 0.002  &   $-$0.003  &   $-$0.002&&  5e$-$04  &   0.0014  &   0.0063 \\
			&$-$0.5  & $-$0.505  &   1  &   1     &  4 & 4&&  0.005  &    0.000  &   0.000 && 7e$-$04  &   0.0012  &   0.0226 \\
			$1000$&$-$0.5  & $-$0.501  &   0  &   0.009 &  0 & 0.017&& 0.001  &   $-$0.009  &   $-$0.017&&0.001  &  0.0048  &    0.0114 \\
			&$-$0.5  & $-$0.505  &   0  &   0.005 &  2 & 2.01&& 0.005  &   $-$0.005  &   $-$0.010&&3e$-$04  &  0.0013  &    0.0065  \\
			&$-$0.5  & $-$0.505  &   0  &   0.004 &  4 & 4.005&&  0.005  &   $-$0.004  &   $-$0.005&&4e$-$04  &  0.0011  &    0.0223  \\
			&$-$0.5  & $-$0.505  &  $-$1  &  $-$0.985 &  0 & 0.023&& 0.005  &   $-$0.015  &   $-$0.023&& 4e$-$04  &  0.0049  &   0.0152\\
			&$-$0.5  & $-$0.503  &  $-$1  &  $-$1     &  2 & 2.02&& 0.003  &    0.000  &   $-$0.020 && 2e$-$04  &  0.0026  &   0.0204\\
			&$-$0.5  & $-$0.505  &  $-$1  &  $-$0.996 &  4 & 4.006 && 0.005  &   $-$0.004  &  $-$0.006&&1e$-$04  &  7e$-$04   &   0.0114\\
			\hline
		\end{tabular}\label{Tablexim05}
	}
\end{table}

 \section{Applications to environmental data}
\label{sec:7}
\noindent

In this section, to illustrate the applicability of the model proposed and its advantages, the BGEV distribution was fitted to two data sets of the climate of the Federal District in Brazil. Furthermore, we compared the BGEV distribution with the
GEV distribution. We adopt the MLE method (as discussed in Sect. \ref{mle}) and all the computations were
done using the R software ( R-Team 2020). For the parameters estimates  of the GEV distribution we use  the fExtremes package (Wuertz et al. 2017).

The data series employed were taken from the of INMET (National Institute of Meteorology) at \url{www.inmet.gov.br.
Data Stations - Automatic Stations}. The data here  utilized  correspond to the
wind speed during the period 10/12/2018 to 10/12/2019 and the temperature at the oval point during the period 24/10/2018 to 20/10/2019.
We used the maximum block technique with blocks of size 24 based in daily maximum value.
Furthermore, the transformation $(x - \rm{Mean})/SD$, where SD is the standard deviation of the data, is to used to adjust it in the
BGEV distribution. We applied Ljung–Box’s test (Trapletti 2016) 
for verifying the
null hypothesis of serial independence of the data sets. The test statistics did not reject the null hypothesis at the significance level of $1\%$.
The descriptive statistics are given in Table \ref{statdvv}.
Finally, from Fig. \ref{Figure:hist1}, we can see that the data set has bimodality. More information about the applications can be found in 
Paiva (2020)

\begin{figure}[!htbp]
	\caption{ } 
	\begin{center}
		\vspace{-1.43cm}	
		\subfigure[Temperature data set]{\includegraphics[scale=0.6]{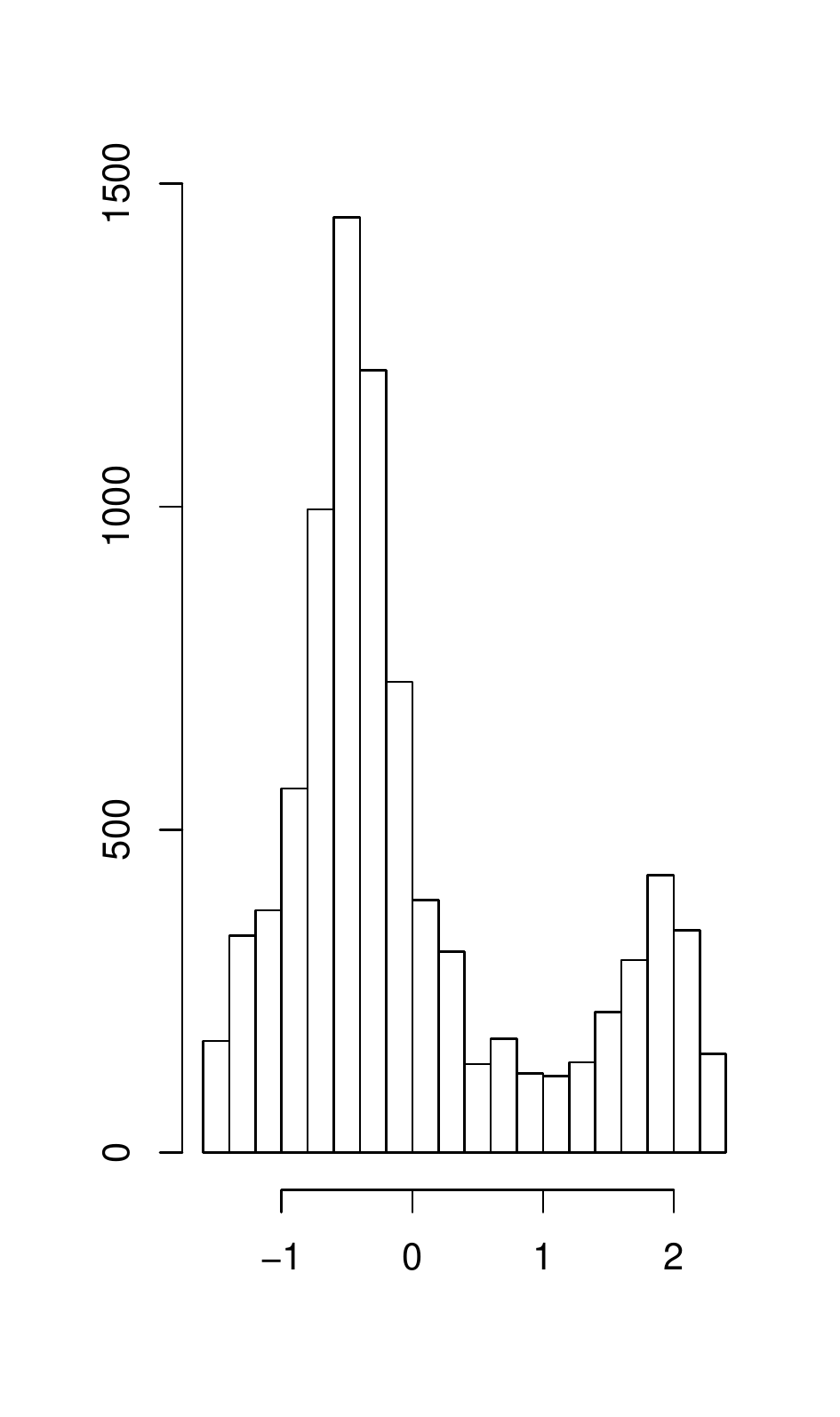}}%
		\subfigure[Wind speed data set]{\includegraphics[scale=0.6]{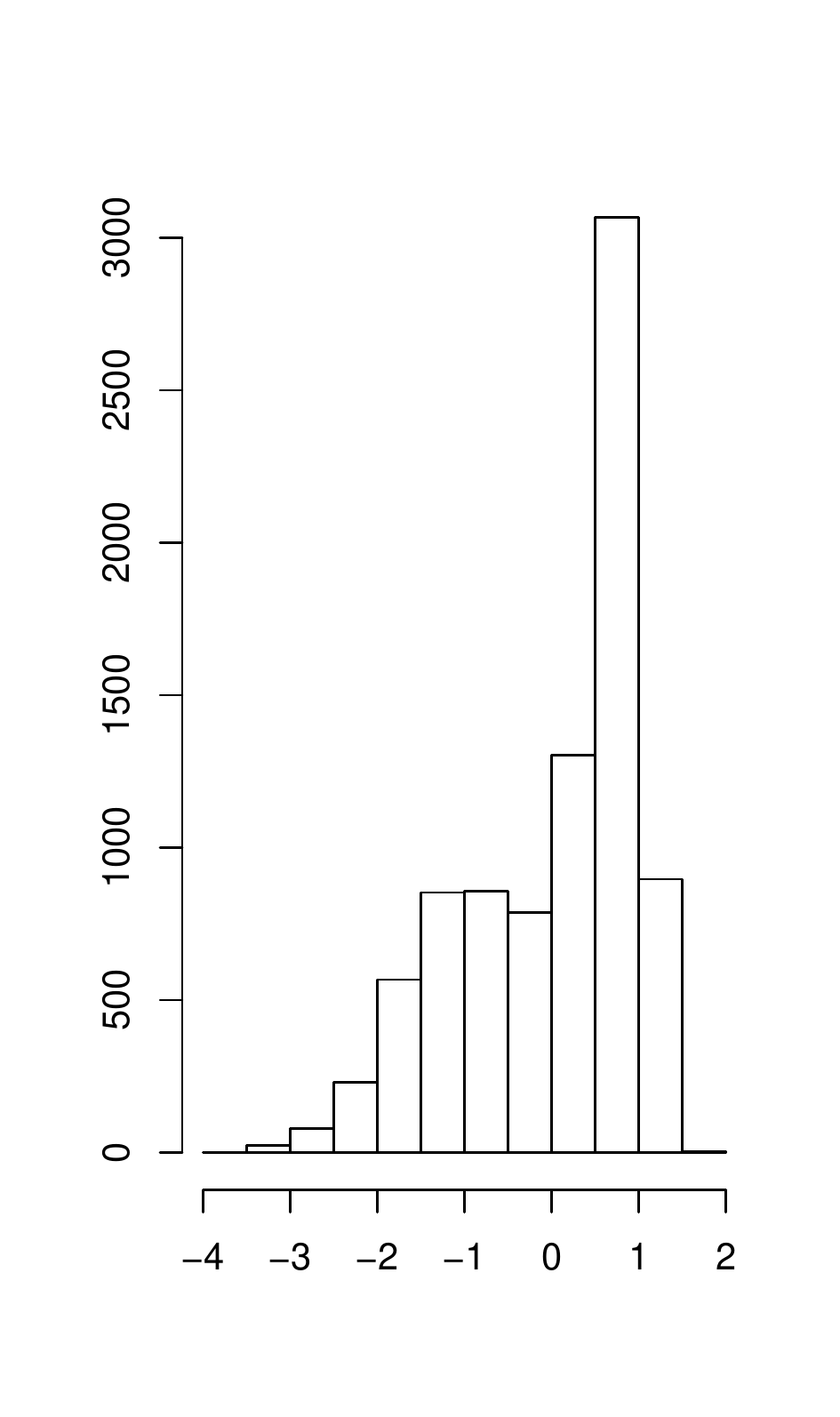}}%
	\end{center}	
	{Histogram of the two data sets.}
	\label{Figure:hist1}
\end{figure}
%

 \begin{table}[!htbp]
 	\centering
 	\caption{Descriptive statistics for the  standardized data sets.}
 \resizebox{\linewidth}{!}{
 	\begin{tabular}{crrrrrr}
 		\hline
 		& Minimum & 1º Quartile & Median & Mean  & 3º Quartile & Maximum  \\
 		\hline
 		Wind speed & $-$2.26 & $-$1.40 & $-$1.08 & $-$0.75 & $-$0.38 & 1.54  \\
 		Temperature at the oval point & $-$3.44& $-$0.43 &0.71 &0.32 & 1.12 &1.88 \\
 		\hline
 	\end{tabular}
 }
 	\label{statdvv}
 \end{table}

The MLEs of the parameters, the value of $-2\ell(\Theta)$, the Kolmogorov-Smirnov (KS) and Anderson-Darling (AD) statistics for the BGEV and GEV models are listed in Table \ref{tabaplic}. In general, the smaller the values of these statistics,
the better the fit. The graphical study performed in Sect. \ref{sec:4}  helped in choosing  of the initial values of the estimates. For the wind speed the initial values were   $\mu = 0, \sigma = 1,  \xi = -0.5$ and $\delta = 0.5$ and for the temperature at the oval point were $\mu = 0, \sigma = 1, \xi = -0.25 $ and $\delta= 0.5$. Since the values are smaller for the BGEV distribution compared with those values of the other
models, the new distribution seems to be a very competitive model for these data sets.
These results illustrate the potentiality of the BGEV model and the importance of the additional parameter.

\begin{table}[!htbp]
\begin{center} \caption{MLEs of the model parameters and the statistics KS, AD and $-2\ell(\Theta)$ for the two data sets.}\label{tabaplic} \vspace{0.2cm}
\begin{tabular}{lrrrrrrrr}
\hline
\multicolumn{5}{c}{Estimates}&&\multicolumn{3}{c}{Statistic}\\
\cline{2-5} \cline{7-9}
Wind speed   &$\widehat\mu$      &$\widehat\sigma$       &$\widehat\xi$       &$\widehat\delta$          & &KS &AD  & $-2\ell(\Theta)$\\
\hline
BGEV    & $-$1.0430  & 0.6516  &  $-$0.3042  &   0.9434    & &  0.05396   &  47.802& 22204.9    \\\\
GEV     & $-$0.4957  & 0.6853  &  0.1330     &  0     & & 0.24109    & 45.179 &  22214.6   \\
\hline
Temperature  &$\widehat\mu$      &$\widehat\sigma$       &$\widehat\xi$       &$\widehat\delta$          & &KS &AD  & $-2\ell(\Theta)$\\
\hline
BGEV    & 0.1804  &  0.8197 & $-$0.4839   & 0.4201      & & 0.0738    & 66.297 &    \\\\
GEV     & $-$0.1955  &  1.0599 & $-$0.6045   &  0     & &  0.3261   & 70.706 &  22111.0  \\
\hline
\end{tabular}
\end{center}
\end{table}

The fitted density (for the BGEV distribution) for the first and second data sets are
displayed in Figures \ref{fitvv} and \ref{fitpov} (together with the QQ plots), respectively.
These results illustrate the potentiality of the BGEV distribution and
the importance of the  additional parameter.
\begin{figure}[H]
	\caption{ } 	
	\begin{center}
		\vspace{-3.5cm}
		\includegraphics[width=0.44\linewidth]{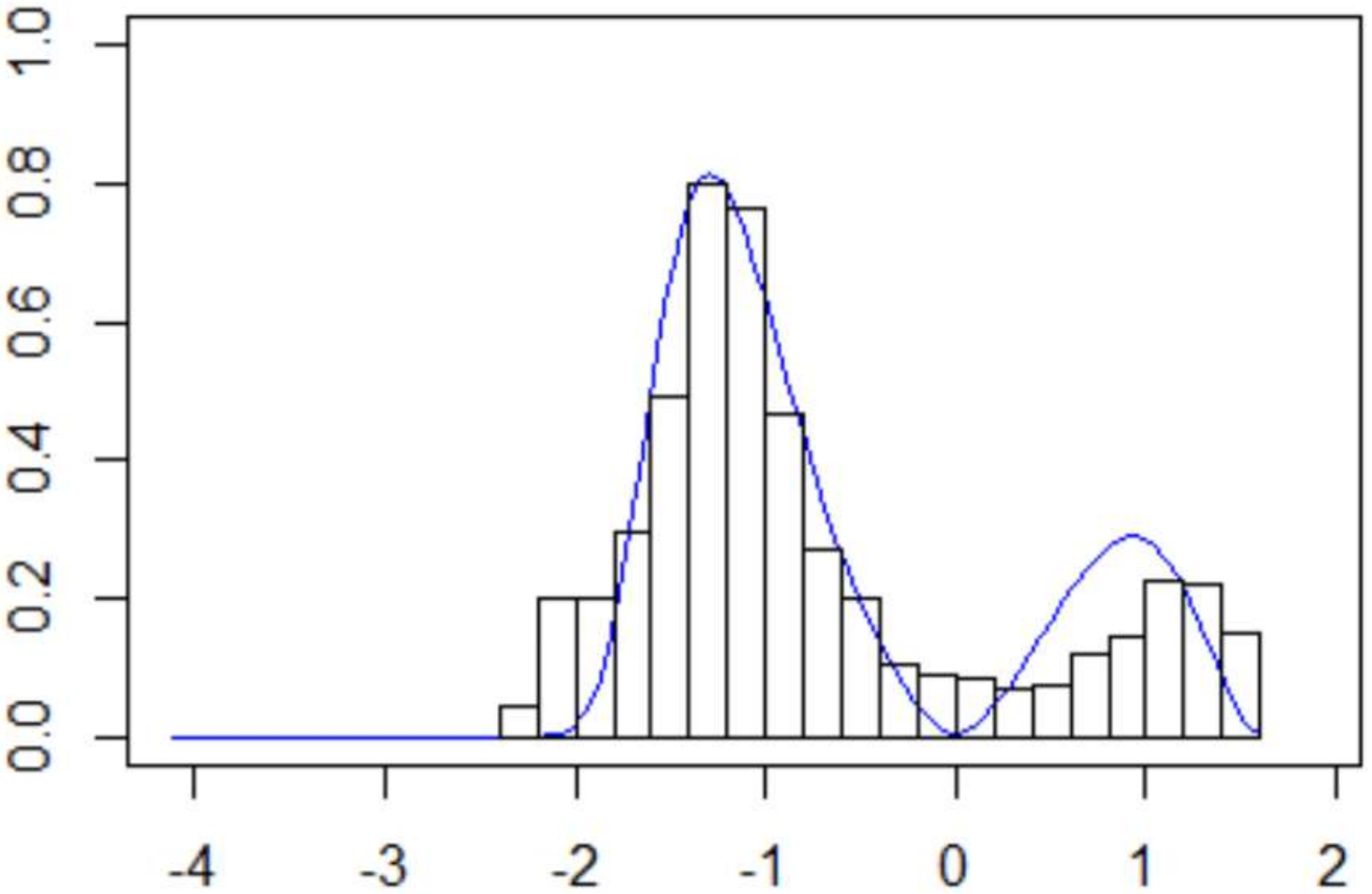}
		\includegraphics[width=0.44\linewidth]{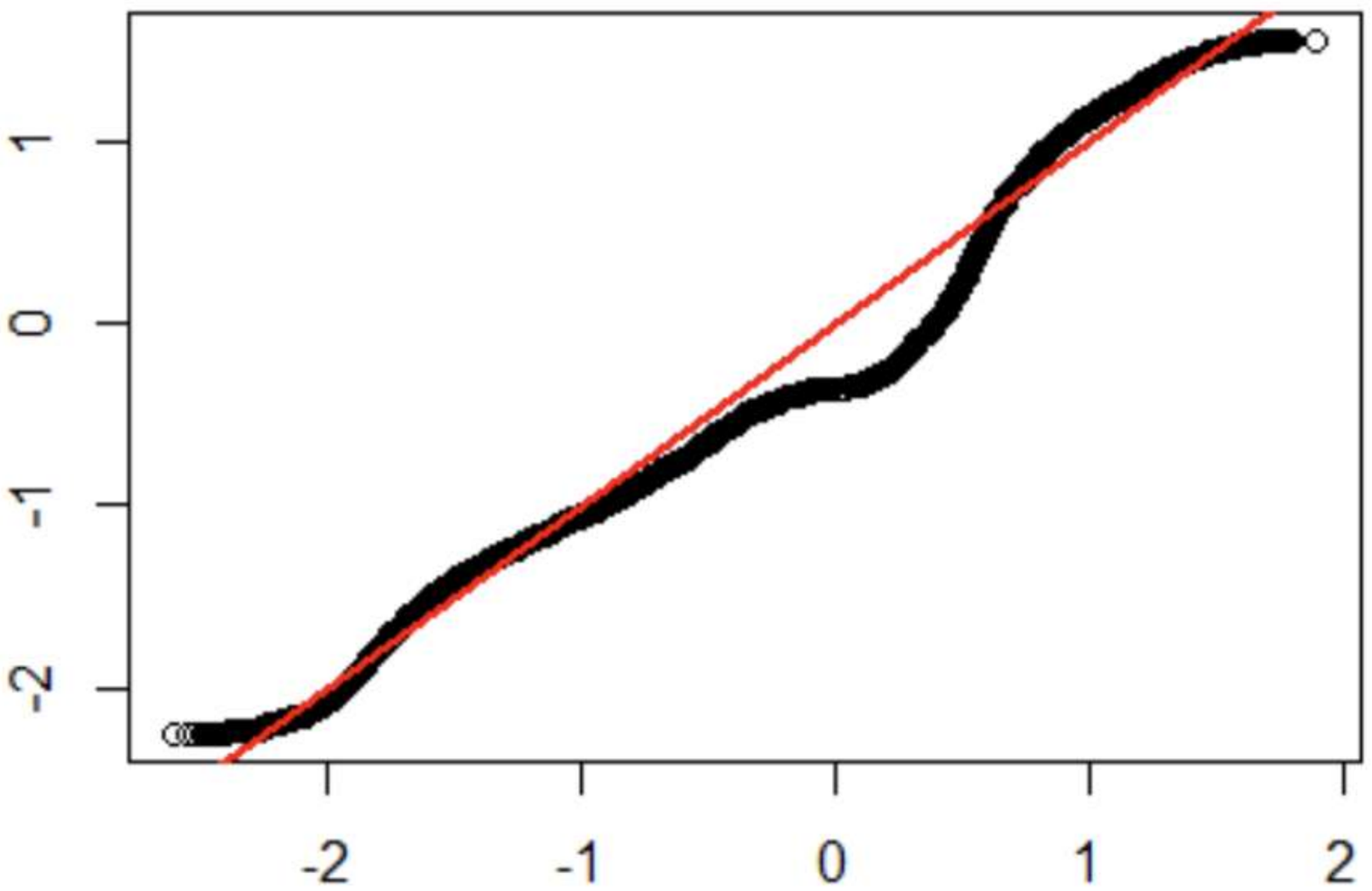}
		\vspace{-2.8cm}
	\end{center} 
	{Histogram with estimated PDF (first panel) and QQ plots (second panel)
		for the BGEV distribution based on the Wind speed data set.}
	\label{fitvv}
\end{figure}
\begin{figure}[H]
	\caption{ } 	
	\begin{center}
				\vspace{-3.5cm}	
		\includegraphics[width=0.4\linewidth]{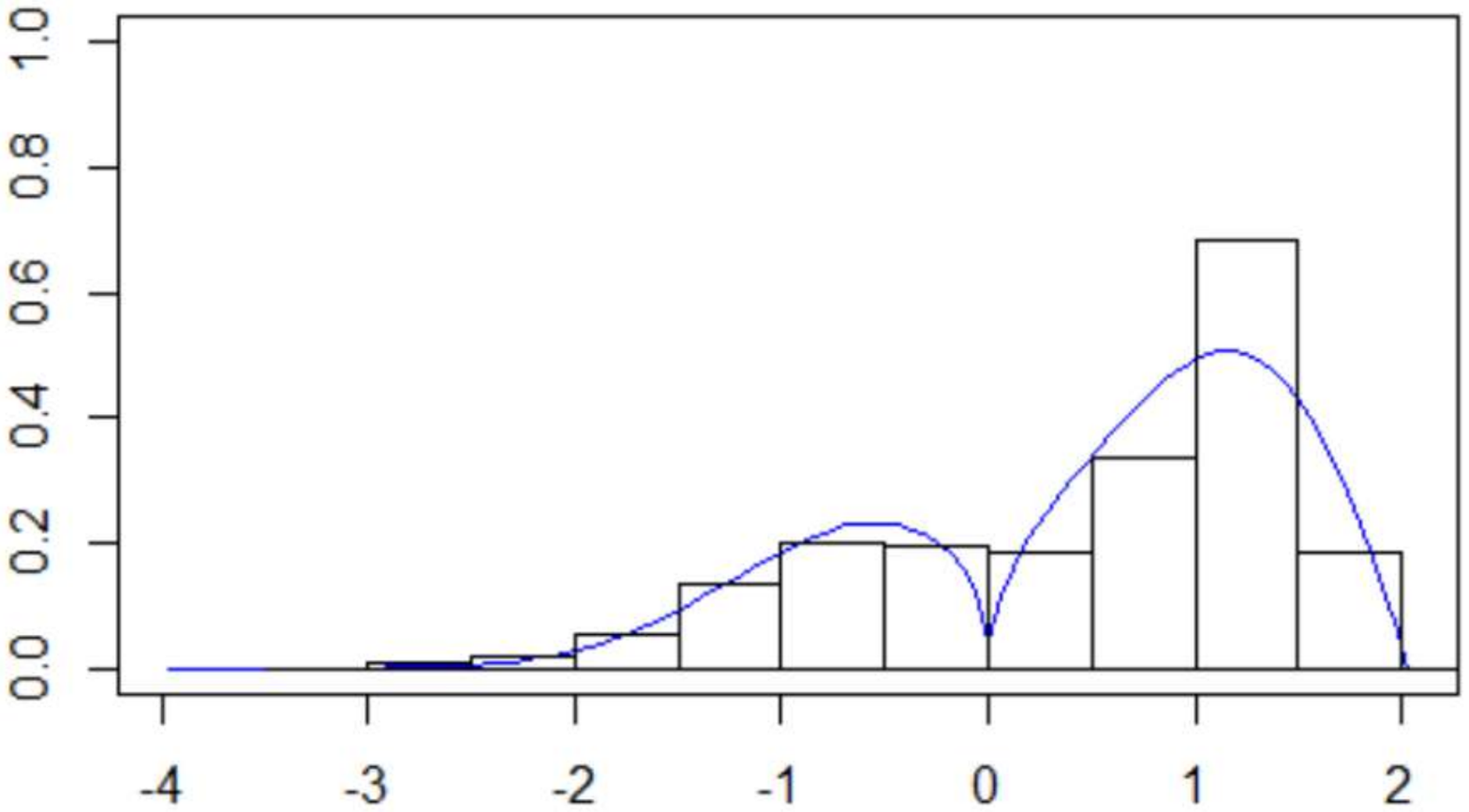}
	\includegraphics[width=0.4\linewidth]{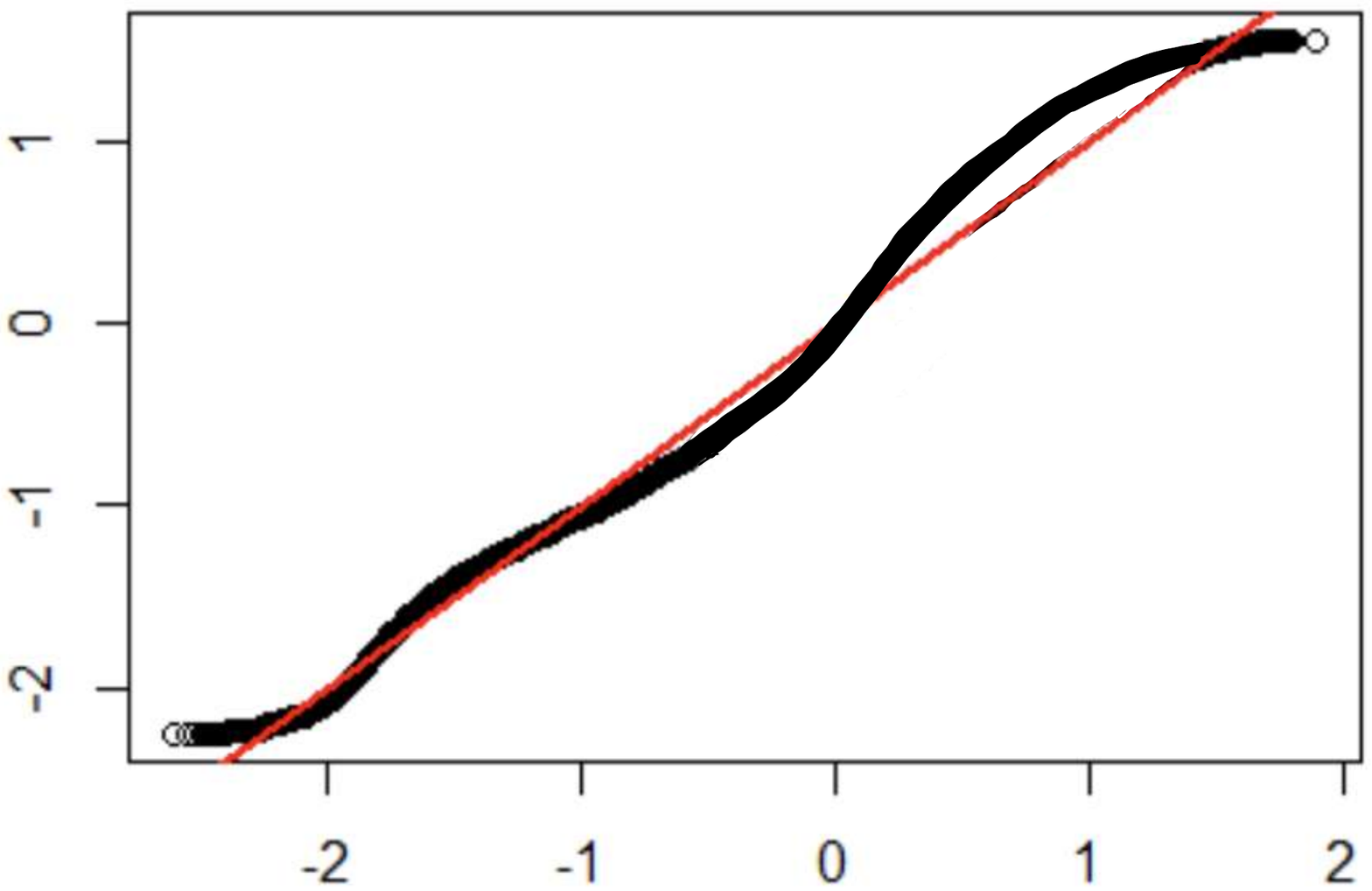}
	\vspace{-2.8cm}
	\end{center}
	{Histogram with estimated PDF (first panel) and QQ plots (second panel)
		for the BGEV distribution based on the Temperature data set.}
	\label{fitpov}
\end{figure}

\section{Conclusions}
\label{sec:8}
\noindent

When extreme value data show bimodality, despite its broad sense applicability
in many fields, the GEV distribution is not suitable.
In this article, we propose a generalization of the GEV distribution, called BGEV distribution, with an additional parameter
which modifies the behavior of the distribution, composing as a alternative models for single
maxima events. The GEV distribution appears as a particular case.
We present important properties of the model and through graphical studies we have shown its wide flexibility. The good performance of the MLEs of the  parameters was tested via Monte Carlo simulation. Applications of the BGEV distribution for two extreme data sets show that the new distribution can be used to provide better adjustments than the GEV to model when the data show bimodality.
We hope this new distribution may attract wider applications for extreme values
analysis.


\end{document}

\section{Appendix A}


\begin{table}[!htbp]
	\centering
	\caption{Empirical means, bias and mean squared errors (MSE) of the
estimates of the parameters for $(\xi, \sigma) = (1, 1)$ and some values of $\mu, \delta$ and $n$.}
\resizebox{\linewidth}{!}{
	\begin{tabular}{cccrrcccrrrrrcccc}
		\hline
		&\multicolumn{6}{c}{Empirical means} && \multicolumn{3}{c}{Bias} && \multicolumn{3}{c}{MSE} \\
		$n$&$\xi$ &  $\widehat{\xi}$  & $\mu$ & $\widehat{\mu}$ &  $\delta$ & $\widehat{\delta}$  && $\widehat{\xi}$ & $\widehat{\mu}$ &  $\delta$  && $\widehat{\xi}$ & $\widehat{\mu}$ &  $\widehat{\delta}$  \\ \hline
		&1 & 1.018& $-$1 &$-$1.017 & 0 & 0.036 && $-$0.018& 0.017&$-$0.036&& 0.049& 0.0214& 0.0171\\
		&1 & 1.038& &$-$1.013& 2 &  2.082 && $-$0.038& 0.013& $-$0.082&& 0.039& 0.0206& 0.1167\\
		&1 &1.034 &$-$1& $-$1.053& 4 & 4.292	&& $-$0.034& 0.053& $-$0.292&& 0.0449& 0.0201 &0.3219\\
50		&1 &1.085 & 0 &0.007 & 0 &0.048	&&$-$0.085 &$-$0.007 &$-$0.048 &&0.0527 &0.0081& 0.0229\\
		&1 &1.085 &0 &0.009 &2 &2.126 &&$-$0.085 &$-$0.009 &$-$0.126 &&0.0433 &0.007& 0.1972 \\
		&1 &1.049 & 0 &0.028 & 4 &4.152 	&&$-$0.049 &$-$0.028 &$-$0.152 &&0.0546 &0.0078& 0.4478\\
		&1 &1.148 &1 &0.97  &0 &0.096 &&$-$0.148 &0.03 &$-$0.096 &&0.0797 &0.029 &   0.0668\\
		&1 &1.083 &1 & 0.995 & 2 &2.14 &&$-$0.083 & 0.005 &$-$0.14 && 0.0867 & 0.0279& 0.4675\\
		&1 &1.125 &1 &0.973 & 4 &4.452	 &&$-$0.125 & 0.027 &$-$0.452&& 0.0954& 0.0295& 2.3334\\ \hline
		&1 & 1.018& $-$1 &$-$1.002& 0 &0.009&& $-$0.018& 0.002& $-$0.009&& 0.0195& 0.0073& 0.0069\\
		&1 & 1.037& $-$1 &$-$1.007& 2 &2.025&& $-$0.037& 0.007& $-$0.025&& 0.0151& 0.0103& 0.0635\\
		&1 & 1.022&$-$1 &$-$1.007 &4 &4.114 &&$-$0.022 &0.007 &$-$0.114 &&0.019 &0.0092& 0.1591\\
100		&1 & 1.038& 0 &0.005  &0  &0.024&& $-$0.038& $-$0.005& $-$0.024&& 0.0189& 0.0027& 0.0086\\
		&1 & 1.07 & 0 &$-$0.005 &2 &2.079 &&$-$0.07 &0.005 &$-$0.079 &&0.021& 0.0034& 0.089\\
		&1 & 1.029& 0 &0.005  &4  &4.079 && $-$0.029& $-$0.005& $-$0.079&& 0.0177 &0.0039& 0.2877\\
		&1 & 1.08 & 1 &0.97   &0  &0.071 &&$-$0.08& 0.03& $-$0.071&& 0.034& 0.0142& 0.0295\\
		&1 & 1.076& 1 &0.978  &2  &2.199 &&$-$0.076& 0.022& $-$0.199&& 0.037& 0.0124& 0.2379\\
		&1 & 1.036& 1 &0.995  &4  &4.199 &&0.036& 0.052& $-$0.159&& 0.037& 0.0124& 0.0237\\ \hline
		&1&  0.999& $-$1&$-$1.003 & 0& 0.003 &&  0.013& 0.001 &0.003  &&$-$0.013 &0.0048 &0.0038\\
		&1&  1.006& $-$1& $-$1.002& 2&  2.026&& $-$0.006& 0.002 &$-$0.026 &&0.0046 &0.0029 &0.0243\\
		&1&  1.018& $-$1& $-$1.01 & 4& 4.051 &&$-$0.018 &0.01   &$-$0.051 &&0.0069 &0.0036 &0.0744\\
250		&1&  1.028& 0 & $-$0.001& 0& 0.019 &&$-$0.028 &0.001  &$-$0.019 &&0.0065 &0.0011 &0.0042\\
		&1&  1.031& 0 & 0.001 & 2& 2.043 &&$-$0.031 &$-$0.001 &$-$0.043 &&0.0078 &0.0012 &0.0318\\
		&1&  1.016& 0 & 0.001 & 4& 4.065 &&$-$0.016 &$-$0.001 &$-$0.065 &&0.0068 &0.0014 &0.1125\\
		&1&  1.044& 1 & 0.983 & 4& 0.039 &&$-$0.044 &0.017  &$-$0.039 &&0.0181 &0.0079 &0.0165\\
		&1&  1.009& 1 & 1.005 & 2& 2.045 &&$-$0.009 &$-$0.005 &$-$0.045 &&0.0138 &0.0067 &0.1171\\
		&1&  1.031& 1 & 0.991 & 4& 4.059 &&$-$0.031 &0.009  &$-$0.059 &&0.0125 &0.0057 &0.2744\\ \hline
		&1    &  1.005  &   1  &   0.996 &  0 & 0.012 &&  $-$0.005  &    0.004  &   $-$0.012&&   0.0028 &   0.0014  &   0.0025      \\
		&1    &  1.011  &   1  &   0.995 &  2 & 2.024 &&  $-$0.011  &    0.005  &  $-$0.024&&  0.0025 &   0.0014  &   0.0213     \\
		&1    &  1.002  &   1  &   1.001 &  4 & 3.998 && $-$0.002  &    $-$0.001  &    0.002&&  0.0029 &   0.0014  &   0.0578      \\
1000	&1    &  1.01   &   0  &   0     &  0 & 0.005 &&  $-$0.010  &    0.000  &   $-$0.005&&   0.002  &    4e$-$04   &  0.0010      \\
		&1    &  1.011  &   0  &  $-$0.003 &  2 & 2.00  &&  $-$0.011  &     0.003  &   $-$0.001&&  0.0019 &    4e$-$04   &  0.0064       \\
		&1    &  1.008  &   0  &   0.001 &  4 & 4.039 && $-$0.008  &    $-$0.001  &  $-$0.039&&  0.0015 &    4e$-$04   &   0.0216       \\
		&1    &  1.001  &  $-$1  &  $-$1.006 &  0 & 0.006 &&   $-$0.001  &    0.006  &   $-$0.006&&  0.0019 &    9e$-$04   &   6e$-$040     \\
		&1    &  0.998  &  $-$1  &  $-$1.002 &  2 & 2.019 &&   0.002  &    0.002  &  $-$0.019&& 0.0011 &    9e$-$04   &  0.0059        \\
		&1    &  1.001  &  $-$1  &  $-$1.001 &  4 & 4.012 && $-$0.001  &    0.001  &   $-$0.012&& 0.0014 &   6e$-$04   &   0.0140    \\	\hline
	\end{tabular}\label{Tablexi1}
}
\end{table}

\begin{table}[!htbp]
	\centering
	\caption{Empirical means, bias and mean squared errors (MSE) of the
estimates of the parameters for $(\xi, \sigma) = (0.5, 1)$ and some values of $\mu, \delta$ and $n$.}
\resizebox{\linewidth}{!}{
	\begin{tabular}{cccrrcccrrrrrcccc}
		\hline
		&\multicolumn{6}{c}{Empirical means} && \multicolumn{3}{c}{Bias} && \multicolumn{3}{c}{MSE} \\
		
		$n$&$\xi$ &  $\widehat{\xi}$  & $\mu$ & $\widehat{\mu}$ &  $\delta$ & $\widehat{\delta}$  && $\widehat{\xi}$ & $\widehat{\mu}$ &  $\delta$  && $\widehat{\xi}$ & $\widehat{\mu}$ &  $\widehat{\delta}$  \\ 		\hline
		&0.5&  0.501 &$-$1& $-$0.999& 0& 0.032&& $-$0.001& $-$0.001& $-$0.032&& 0.0233& 0.024& 0.0129\\
		&0.5&  0.524 &$-$1& $-$1.032& 2& 2.15 &&$-$0.024 &0.032 &$-$0.15 &&0.0268& 0.025& 0.1033\\
		&0.5&  0.486 &$-$1& $-$1.019& 4& 4.215&& 0.014 &0.019 &$-$0.215 &&0.0254& 0.023& 0.2411\\
	50	&0.5&  0.549 &0 & 0.064 & 0& $-$0.049&& 0.028& $-$0.064& 0.0375&& 0.0154& 0.0282 &0321\\
		&0.5&  0.562 &0 & 0.015 & 2& 2.118 &&$-$0.062& $-$0.015& $-$0.118&& 0.045& 0.0175& 0.1552\\
		&0.5&  0.551 &0 & 0.019 & 4& 4.282 &&$-$0.051& $-$0.019& $-$0.282&& 0.0303& 0.0124& 0.4835\\
		&0.5&  0.541 &1 & 1.029 & 0& 0.057 &&$-$0.041& $-$0.029& $-$0.057&& 0.0725& 0. 0255& 0.0515\\
		&0.5&  0.559 &1 & 1.006 & 2& 2.094 &&$-$0.059& $-$0.006& $-$0.094&& 0.0786& 0.0254& 0.3663\\
		&0.5&  0.563 &1 & 0.998 & 4& 4.379 &&$-$0.063& 0.002 &$-$0.379 &&0.083& 0.0238 &1.7225\\ \hline
		&0.5&  0.503& $-$1& $-$1.011& 0& 0.024&& $-$0.003& 0.011& $-$0.024 &&0.0113& 0.0142& 0.0068\\
		&0.5&  0.503& $-$1& $-$1.02 & 2& 2.09 && $-$0.003& 0.02 &$-$0.09 &&0.0088 &0.0128& 0.0469\\
		&0.5&  0.518& $-$1& $-$1.017& 4& 4.145&& $-$0.018& 0.017& $-$0.145&& 0.0094& 0.011& 0.1269\\
	100	&0.5& 0.525 & 0 &  0.005& 0& 0.025&& $-$0.025& $-$0.005& $-$0.025&& 0.0115& 0.0052& 0.0084\\
		&0.5& 0.511 & 0 &  0.021& 2& 2.067&& $-$0.011& $-$0.021& $-$0.067&& 0.0105& 0.0065& 0.0663\\
		&0.5& 0.52  & 0 &  0.012& 4& 4.136&& $-$0.02 &$-$0.012 &$-$0.136 &&0.0131 &0.0057& 0.2212\\
		&0.5& 0.53  & 1 & 1.012 & 0& 0.011&& $-$0.03 &$-$0.012 &$-$0.011 &&0.0281 &0.01 &0.0195\\
		&0.5& 0.53  & 1 & 0.997 & 2& 2.102&& $-$0.03 &0.003 &$-$0.102 &&0.0519 &0.0125 &0.3862\\
		&0.5& 0.516 & 1 & 1.002 & 4& 4.093&& $-$0.016& $-$0.002& $-$0.093&& 0.024 &0.0084& 0.3418\\ 		\hline
		&0.5&  0.513&  $-$1& $-$1.002&0& 0.011&& 0.002 &$-$0.011& 0.0042&& 0.0033& 0.002& 0.232\\
		&0.5&  0.505& $-$1 &$-$1.009 &2& 2.027&& $-$0.005& 0.009& $-$0.027&& 0.0051& 0.004 &0.0164\\
		&0.5&  0.501& $-$1 &$-$1.009 &4& 4.044&& $-$0.001& 0.009& $-$0.044&& 0.0042& 0.0054& 0.0571\\
250		&0.5&  0.511&  0 & 0.009 &0& 0.01 &&$-$0.011 &$-$0.009& $-$0.01 &&0.0033 &0.0024 & 0.0035\\
		&0.5&  0.505&  0 & 0.008 &2& 2.035&& $-$0.005& $-$0.008& $-$0.035&& 0.0053& 0.0026& 0.0315\\
		&0.5&  0.514&  0 &$-$0.005 &4& 4.051&& $-$0.014& 0.005 &$-$0.051&& 0.006 &0.003 &0.0654\\
		&0.5&  0.537&  1 & 0.988 &0& 0.023&& $-$0.037& 0.012 &$-$0.023&& 0.0108 &0.0045& 0.0066\\
		&0.5&  0.511&  1 & 1.005 &2& 2.008&& $-$0.011& $-$0.005& $-$0.008&& 0.0109& 0.0047& 0.0639\\
		&0.5&  0.528&  1 & 0.998 &4& 4.062&& $-$0.028& 0.002 &$-$0.062&& 0.0108 &0.0046 &0.1906\\ \hline
		&0.5  &  0.493  &   1  &   1.003 &  0 &$-$0.003&& 0.007  &   $-$0.003  &    0.003&&0.0026 &    0.0015  &   0.0020 \\
		&0.5  &  0.512  &   1  &   0.992 &  2 & 2.041&&  $-$0.012  &    0.008  &  $-$0.041&& 0.0023 &    0.0013  &   0.0167 \\
		&0.5  &  0.506  &   1  &   0.999 &  4 & 4.023&&$-$0.006  &    0.001  &   $-$0.023 && 0.0025 &    0.0011  &    0.0497\\
1000	&0.5  &  0.507  &   0  &  $-$0.005 &  0 &$-$0.001&& $-$0.007  &    0.005  &    0.001 &&8e$-$04  &    6e$-$04   &    9e$-$04 \\
		&0.5  &  0.501  &   0  &   0.003 &  2 & 2.006&& $-$0.001  &    $-$0.003 &  $-$0.006 && 0.0013 &    7e$-$04   &    0.0061 \\
		&0.5  &  0.505  &   0  &  $-$0.001 &  4 & 3.998&&  $-$0.005  &    0.001  &   0.002 &&9e$-$04  &    7e$-$04   &    0.0135 \\
		&0.5  &  0.501  &  $-$1  &  $-$0.998 &  0 & 0    && $-$0.001  &   $-$0.002  &    0.000 && 0.0011 &    0.0012  &   6e$-$04  \\
		&0.5  &  0.505  &  $-$1  &  $-$1.001 &  2 & 2.004 && $-$0.005  &   0.001   &   $-$0.004 &&  9e$-$04  &    9e$-$04   &   0.0045 \\
		&0.5  &  0.499  &  $-$1  &  $-$0.997 &  4 & 4.007 && 0.001   &  $-$0.003   &   $-$0.007 &&  0.001  &   0.0013  &    0.0149\\
		\hline
	\end{tabular}\label{Tablexi05}
}
\end{table}


\begin{table}[!htbp]
	\centering
	\caption{Empirical means, bias and mean squared errors (MSE) of the
estimates of the parameters for $(\xi, \sigma) = (0.25, 1)$ and some values of $\mu, \delta$ and $n$.}
\resizebox{\linewidth}{!}{
	\begin{tabular}{cccrrcccrrrrrcccc}
		\hline
		&\multicolumn{6}{c}{Empirical means} && \multicolumn{3}{c}{Bias} && \multicolumn{3}{c}{MSE} \\
		
		$n$&$\xi$ &  $\widehat{\xi}$  & $\mu$ & $\widehat{\mu}$ &  $\delta$ & $\widehat{\delta}$  && $\widehat{\xi}$ & $\widehat{\mu}$ &  $\delta$  && $\widehat{\xi}$ & $\widehat{\mu}$ &  $\widehat{\delta}$  \\		\hline
		&0.25&  0.256& $-$1& $-$1.008& 0 & 0.04 && $-$0.045& 0.128& $-$0.04&& 0.0216& 0.0322& 0.0116\\
		&0.25&  0.244& $-$1& $-$1.018& 2 & 2.116&& 0.006 & 0.018& $-$0.116&& 0.0162& 0.0323& 0.0838\\
		&0.25&  0.239& $-$1& $-$1.013& 4 & 4.185&& 0.011 & 0.013& $-$0.185 &&0.0186& 0.0262& 0.2052\\
50		&0.25&  0.273&  0&  0.016& 0 & 0.056&& $-$0.023& $-$0.016& $-$0.056&& 0.0302& 0.0155& 0.0251\\
		&0.25&  0.285&  0& 0.026 & 2 & 2.131&& $-$0.035& $-$0.026& $-$0.131&& 0.035 &0.0168 &0.2199\\
		&0.25&  0.275&  0& 0.02  & 4 & 4.213&& $-$0.025& $-$0.02 &$-$0.213 &&0.0247 &0.0143 &0.4635\\
		&0.25&  0.267&  1& 1.001 & 0 & 0.037&& $-$0.017& $-$0.001& $-$0.037&& 0.0803& 0.029 &0.0414\\
		&0.25&  0.289&  1& 1.017 & 2 &2.197 &&$-$0.039 &$-$0.017 &$-$0.197 &&0.0827 &0.0283 &0.5167\\
		&0.25&  0.292&  1& 1.009 & 4 &4.309 &&$-$0.042 &$-$0.009 &$-$0.309 &&0.0812 &0.0226 &1.3124	\\ \hline
		&0.25& 0.245& $-$1 &$-$1.021& 0& 0.037&& 0.005& 0.021& $-$0.037 &&0.0096& 0.0199& 0.0054\\
		&0.25& 0.247& $-$1 &$-$1.007& 2& 2.071&& 0.003& 0.007& $-$0.071 &&0.0075& 0.0157& 0.048\\
		&0.25& 0.241& $-$1 &$-$1.004& 4& 4.045&& 0.009& 0.004& $-$0.045 &&0.007 &0.0114 &0.1047\\
100		&0.25& 0.255&  0 & 0.027& 0& 0.028&& $-$0.005& $-$0.027& $-$0.028&& 0.0104& 0.0074& 0.007\\
		&0.25& 0.262&  0 & 0.01 & 2& 2.067&& $-$0.012& $-$0.01& $-$0.067 &&0.0107 &0.0087& 0.0579\\
		&0.25& 0.263&  0 & 0.007& 4& 4.115&& $-$0.013& $-$0.007& $-$0.115&& 0.0101& 0.0085& 0.1793\\
		&0.25& 0.249&  1 & 1.009& 0& 0.011&& 0.001 &$-$0.009 &$-$0.011 &&0.025 &0.0108 &0.0124\\
		&0.25& 0.272&  1 & 1.098& 2& 2.058&& $-$0.022&  $-$0.058& 0.0325&& 0.0117& 0.1798& 0.234\\
		&0.25& 0.241&  1 & 1.021& 4& 4.052&& 0.009 &$-$0.021& $-$0.052&& 0.0232& 0.0104& 0.404\\ \hline
		&0.25&  0.2255& $-$1& $-$0.994& 0& 0.007&& $-$0.005& $-$0.006& $-$0.007 &&0.0022 &0.0043& 0.0017\\
		&0.25&  0.243 & $-$1& $-$1.003& 2& 2.053&& 0.007 &0.003 &$-$0.053&& 0.0026& 0.0046& 0.0186\\
		&0.25&  0.254 & $-$1& $-$1.014& 4& 4.072&& $-$0.004& 0.014 &$-$0.072&& 0.003& 0.0052& 0.0436\\
250		&0.25&  0.267 & 0 & 0.002 & 0& 0.02 &&$-$0.017 &$-$0.002 &$-$0.02 &&0.0028& 0.0033& 0.0032\\
		&0.25&  0.255 & 0 & 0.005 & 2& 2.01 &&$-$0.005 &$-$0.005 &$-$0.01 &&0.0034& 0.0026& 0.0229\\
		&0.25&  0.249 & 0 & 0.001 & 4& 4.038 &&0.001 &$-$0.001 &$-$0.038&& 0.0046& 0.0031& 0.0772\\
		&0.25&  0.236 & 1 & 1.01  & 0& $-$0.003&& 0.014& $-$0.01 &0.003 &&0.0079 &0.005 &0.0038\\
		&0.25&  0.241 & 1 & 1.001 &2 & 1.997 &&0.009 &$-$0.001 &0.003 &&0.0068 &0.0037& 0.0469\\
		&0.25&  0.269 & 1 &0.998  &4 & 4.121 &&$-$0.019& 0.002 &$-$0.121 &&0.0087& 0.0049& 0.195\\ 		\hline
		&0.25 &  0.253  &   1  &   1.002 &  0 & 0.002&&$-$0.003  &  $-$0.002  &  $-$0.002&& 0.0015 &    0.001   &   0.0010 \\
		&0.25 &  0.25   &   1  &   0.996 &  2 & 2 &&   0.000  &    0.004  &  0.000  &&  0.002  &    9e$-$04   &   0.0132   \\
		&0.25 &  0.256  &   1  &   0.996 &  4 & 4.018 && $-$0.006  &    0.004  & $-$0.018		&&  0.0018 &    0.0013  &   0.0331 \\
1000	&0.25 &  0.254  &   0  &  $-$0.001 &  0 & 0.002&& $-$0.004  &    0.001  &  $-$0.002 && 9e$-$04  &   8e$-$04   &    6e$-$04 \\
		&0.25 &  0.249  &   0  &   0.004 &  2 & 2.013 && 0.001  &   $-$0.004  &  $-$0.013 && 8e$-$04  &   8e$-$04   &   0.0066\\
		&0.25 &  0.252  &   0  &  $-$0.002 &  4 & 4.012&& $-$0.002  &    0.002  & $-$0.012  && 8e$-$04  &   6e$-$04   &   0.0199\\
		&0.25 &  0.248  &  $-$1  &  $-$1.001 &  0 & 0.002 && 0.002  &     0.001  & $-$0.002 && 6e$-$04  &   0.0014  &   5e$-$04\\
		&0.25 &  0.25   &  $-$1  &  $-$0.998 &  2 & 2.004 && 0.000  &    $-$0.002  & $-$0.004 && 7e$-$04  &   0.0011  &   0.0037\\
		&0.25 &  0.253  &  $-$1  &  $-$1.002 &  4 & 4.015 && $-$0.003  &    0.002  &  $-$0.015 &&  7e$-$04  &   0.001   &   0.0106\\
		\hline
	\end{tabular}\label{Tablexi025}
}
\end{table}


\begin{table}[!htbp]
	\centering
	\caption{Empirical means, bias and mean squared errors (MSE) of the
estimates of the parameters for $(\xi, \sigma) = (-0.25, 1)$ and some values of $\mu, \delta$ and $n$.}
\resizebox{\linewidth}{!}{
	\begin{tabular}{cccrrcccrrrrrcccc}
		\hline
		&\multicolumn{6}{c}{Empirical means} && \multicolumn{3}{c}{Bias} && \multicolumn{3}{c}{MSE} \\
		
$n$		&		$\xi$ &  $\widehat{\xi}$  & $\mu$ & $\widehat{\mu}$ &  $\delta$ & $\widehat{\delta}$  && $\widehat{\xi}$ & $\widehat{\mu}$ &  $\delta$  && $\widehat{\xi}$ & $\widehat{\mu}$ &  $\widehat{\delta}$  \\
		\hline
		&$-$0.25&  $-$0.266& $-$1& $-$0.981& 0& 0.055&& 0.016& $-$0.019& $-$0.055&& 0.0071& 0.0269& 0.0197\\
		&$-$0.25&  $-$0.288& $-$1& $-$1.012& 2& 2.151&& 0.038& 0.012& $-$0.151&& 0.0082& 0.0279& 0.0849\\
		&$-$0.25&  $-$0.296& $-$1& $-$0.997& 4& 4.288&& 0.046& $-$0.003& $-$0.288&& 0.0093& 0.0259& 0.2736\\
50		&$-$0.25&  $-$0.286&  0& 0.002&  0& 0.024&& 0.036& $-$0.002& $-$0.024&& 0.0159& 0.0298& 0.018\\
		&$-$0.25&  $-$0.266&  0& 0.003&  2& 2.094&& 0.016& $-$0.003& $-$0.094&& 0.0125& 0.0243& 0.1833\\
		&$-$0.25 & $-$0.283&  0& 0.031&  4& 4.079 &&0.033& $-$0.031& $-$0.079&& 0.0163& 0.0252& 0.439\\
		&$-$0.25 & $-$0.256		&1& 1.01 &  0&  0.028&& 0.006& $-$0.01& $-$0.028&& 0.0326& 0.0312& 0.0211\\
		&$-$0.25 & $-$0.278&  1& 1.03 &  2& 2.119 &&0.028& $-$0.03 &$-$0.119&& 0.0332& 0.0224& 0.2411\\
		&$-$0.25 & $-$0.257&  1& 1.037&  4& 4.165 &&0.007& $-$0.037 &$-$0.165 &&0.0374& 0.0359& 0.6435		\\ \hline
		&$-$0.25& $-$0.264& $-$1& $-$1.006& 0& 0.027&& 0.014& 0.006& $-$0.027&& 0.0033& 0.0149& 0.0051\\
		&$-$0.25& $-$0.264& $-$1& $-$0.992& 2& 2.056&& 0.014& $-$0.008& $-$0.056&& 0.0032& 0.0169& 0.0435\\
		&$-$0.25& $-$0.272& $-$1& $-$1.002& 4& 4.156&& 0.022& 0.002& $-$0.156&& 0.0041& 0.0116& 0.134\\
100		&$-$0.25& $-$0.258&  0& 0.002 & 0& 0.028&& 0.008& $-$0.002& $-$0.028&& 0.0053& 0.0138& 0.0095\\
		&$-$0.25& $-$0.271&  0& 0.015 & 2& 2.02 &&0.021 & $-$0.015 &$-$0.02&& 0.0056& 0.014& 0.0709\\
		&$-$0.25& $-$0.255&  0& $-$0.001& 4& 4.095&& 0.005& 0.001 &$-$0.095&& 0.005& 0.0118& 0.2409\\
		&$-$0.25& $-$0.256&  1& 1.018 & 0& 0.004&& 0.006& $-$0.018& $-$0.004&& 0.0115& 0.0118& 0.0093\\
		&$-$0.25& $-$0.275&  1& 1.026 &2 &2.007 &&0.025 & $-$0.026& $-$0.007&& 0.0082& 0.013 &0.0599\\
		&$-$0.25& $-$0.274&  1& 1.015 & 4& 4.02 &&0.024 &$-$0.015 &$-$0.02  &&0.0103& 0.0131& 0.234		\\ \hline
		&$-$0.25& $-$0.261&  $-$1& $-$0.994& 0& 0.012&& 0.011& $-$0.006& $-$0.012 &&9e$-$04& 0.0041& 0.0014\\
		&$-$0.25& $-$0.255&  $-$1& $-$0.985& 2& 2.03&& 0.005 &$-$0.015 &$-$0.03  &&0.0011& 0.0043& 0.0166\\
		&$-$0.25& $-$0.259&  $-$1& $-$1.007& 4& 4.037&& 0.009& 0.007 &$-$0.037 &&0.0012& 0.0058& 0.0458\\
250		&$-$0.25& $-$0.257&   0& 0.012 & 0& 0.007&& 0.007&$-$0.012 &0.0018 &&0.0045& 0.0108&0.0769\\
		&$-$0.25& $-$0.257&   0& 0.014 & 2& 2.03 &&0.007 &$-$0.014 &$-$0.03 &&0.0021& 0.0052& 0.027\\
		&$-$0.25& $-$0.258&   0& 0.014 & 4& 4.023&& 0.008& $-$0.014 &$-$0.023&& 0.0022& 0.0067& 0.1011\\
		&$-$0.25& $-$0.258&   1& 1.012 & 0& 0.006&& 0.008& $-$0.012 &$-$0.006&& 0.0044& 0.0046& 0.0035\\
		&$-$0.25& $-$0.268&   1& 0.997 & 2& 1.974&& 0.018& 0.003& 0.026&& 0.0031& 0.0038& 0.026\\
		&$-$0.25& $-$0.256&   1& 1.007 & 4& 4.003&& 0.006& $-$0.007& $-$0.003&& 0.0035& 0.0345& 0.0042\\ \hline
		&$-$0.25 & $-$0.257  &   1  &   1.003 &  0 &$-$0.002&& 0.007  &   $-$0.003  &   0.002&& 9e$-$04  &    0.0013  &   9e$-$04\\
		&$-$0.25 & $-$0.248  &   1  &   1.003 &  2 & 2.01&& $-$0.002  &   $-$0.003  &   $-$0.010&& 8e$-$04  &   0.0013  &   0.0061 \\
		&$-$0.25 & $-$0.253  &   1  &   1     &  4 & 4.003&& 0.003  &    0.000  &   $-$0.003&&  0.001  &    0.0012  &  0.0228 \\
1000	&$-$0.25 & $-$0.251  &   0  &   0     &  0 & 0.002&& 0.001  &     0.000  &  $-$0.002&& 4e$-$04  &    0.0013  &   7e$-$04 \\
		&$-$0.25 & $-$0.252  &   0  &   0     &  2 & 2.004&& 0.002  &     0.000  &  $-$0.004&& 5e$-$04  &   0.0012  &    0.0087 \\
		&$-$0.25 & $-$0.251  &   0  &   0.001 &  4 & 4.003&& 0.001  &   $-$0.001  &   $-$0.003&& 4e$-$04  &     0.0011  &  0.0176 \\
		&$-$0.25 & $-$0.25   &  $-$1  &  $-$1.003 &  0 & 0.003&& 0.000  &   0.003  &   $-$0.003&&  2e$-$04  &    0.003   &    0.0013 \\
		&$-$0.25 & $-$0.253  &  $-$1  &  $-$1     &  2 & 1.998&& 0.003  &   0.000  &    0.002&& 3e$-$04  &   0.0016  &   0.0033\\
		&$-$0.25 & $-$0.254  &  $-$1  &  $-$0.997 &  4 & 4.021&& 0.004  &   $-$0.003  &  $-$0.021 && 2e$-$04  &   0.0012  &   0.0088\\
		\hline
	\end{tabular}\label{Tablexim025}
}
\end{table}


\begin{table}[!htbp]
	\centering
	\caption{Empirical means, bias and mean squared errors (MSE) of the
estimates of the parameters for $(\xi, \sigma) = (-0.5, 1)$ and some values of $\mu, \delta$ and $n$.}
\resizebox{\linewidth}{!}{
	\begin{tabular}{cccrrcccrrrrrcccc}
		\hline
		&\multicolumn{6}{c}{Empirical means} && \multicolumn{3}{c}{Bias} && \multicolumn{3}{c}{MSE} \\
		
$n$		&$\xi$ &  $\widehat{\xi}$  & $\mu$ & $\widehat{\mu}$ &  $\delta$ & $\widehat{\delta}$  && $\widehat{\xi}$ & $\widehat{\mu}$ &  $\widehat{\delta}$  && $\widehat{\xi}$ & $\widehat{\mu}$ &  $\widehat{\delta}$  \\ \hline
		&$-$0.5& $-$0.553& $-$1& $-$0.955& 0& 0.044&& 0.053& $-$0.045& $-$0.044&& 0.0102& 0.0305& 0.0123\\
		&$-$0.5& $-$0.563& $-$1& $-$0.948& 2& 2.229&& 0.063& $-$0.052& $-$0.229&& 0.0129& 0.0319& 0.1927\\
		&$-$0.5& $-$0.579& $-$1& $-$0.952& 4& 4.25 &&0.079 &$-$0.048 &$-$0.25&& 0.0077& 0.0209& 0.3505\\
$50$	&$-$0.5& $-$0.545&  0&  0.04 & 0& 0.017&& 0.045 &$-$0.04& $-$0.017&& 0.012& 0.0259& 0.0197\\
		&$-$0.5& $-$0.565&  0&  0.032& 2& 2.005&& 0.065 &$-$0.032& $-$0.005&& 0.0112& 0.0224& 0.1354\\
		&$-$0.5& $-$0.553& $-$1&$-$0.553&  0&  0.0434&& 4.069& 0.053& $-$0.043&& $-$0.069& 0.0127& 0.0275\\
		&$-$0.5& $-$0.562&  1&  1.021& 0& $-$0.012&& 0.062 & $-$0.021& 0.012&& 0.0258& 0.0248& 0.0211\\
		&$-$0.5& $-$0.532&  1&  1.025& 2& 2.067 &&0.032& $-$0.025& $-$0.067&& 0.0247& 0.0286& 0.1658\\ 	\hline	
		&$-$0.5&  $-$0.536& $-$1& $-$0.964& 0& 0.025&& 0.036& $-$0.036& $-$0.025&& 0.0036& 0.0131& 0.0052\\
		&$-$0.5&  $-$0.528& $-$1& $-$0.982& 2& 2.076&& 0.028& $-$0.018& $-$0.076&& 0.003& 0.0114& 0.0457\\
		&$-$0.5&  $-$0.531& $-$1& $-$0.984& 4& 4.182&& 0.031& $-$0.016& $-$0.182&& 0.0027& 0.0093& 0.1574\\
$100$	&$-$0.5&  $-$0.523&  0&  0.015& 0& 0.015&& 0.023& $-$0.015& $-$0.015&& 0.004& 0.0113& 0.008\\
		&$-$0.5&  $-$0.532&  0&  0.026& 2& 2.053&& 0.032& $-$0.026& $-$0.053&& 0.0042& 0.011& 0.0643\\
		&$-$0.5&  $-$0.513&  0&  0.008& 4& 4.153&& 0.013& $-$0.008& $-$0.153&& 0.0042& 0.0101& 0.2634\\
		&$-$0.5& $-$0.516 &  1&  1.01 & 0& 0.016&& 0.016 &$-$0.01 &$-$0.016&& 0.011 &0.0113& 0.0102\\
		&$-$0.5& $-$0.532 &  1&  1.008& 2& 1.988&& 0.032 &$-$0.008& 0.012 &&0.0095 &0.0116& 0.0731\\
		&$-$0.5& $-$0.527 &  1&  0.983& 4& 3.955&& 0.027 &0.017 &0.045 &&0.0089 &0.0103& 0.1963		\\ \hline
		&$-$0.5 &$-$0.516& $-$1& $-$0.985&  0& 0.018&& 0.016& $-$0.015& $-$0.018&& 7e$-$04& 0.0033& 0.002\\
		&$-$0.5 &$-$0.516& $-$1& $-$0.982&  2& 2.049&& 0.016& $-$0.018& $-$0.049&& 0.0012& 0.0096& 0.0346\\
		&$-$0.5 &$-$0.51 &$-$1 & $-$1.004&  4& 4.096&& 0.01& 0.004& $-$0.096&& 8e$-$04& 0.0035& 0.0767\\
$250$	&$-$0.5 &$-$0.514& 0 & 0.011 &  0& $-$0.003&& 0.014& $-$0.011& 0.003&& 0.0012& 0.0033& 0.003\\
		&$-$0.5 &$-$0.51 & 0 & 0.007 &  2& 2.025 &&0.01& $-$0.007& $-$0.025&& 0.0015& 0.0042& 0.0352\\
		&$-$0.5 &$-$0.509& 0 &$-$0.001 &  4& 4.009 && 0.009& 0.001& $-$0.009&& 0.0012& 0.0038& 0.0823\\
		&$-$0.5 &$-$0.51 & 1 &1.011  &  0& 0.006 &&0.01& $-$0.011& $-$0.006&& 0.0031& 0.0068& 0.0031\\
		&$-$0.5 &$-$0.513& 1 &0.999  &  2& 1.99 &&0.013& 0.001& 0.01&& 0.0028& 0.004& 0.0283\\
		&$-$0.5 &$-$0.511& 1 &1.007  &  4& 3.995 &&0.011& $-$0.007& 0.005&& 0.0033& 0.0061& 0.0712\\ \hline
		&$-$0.5  & $-$0.505  &   1  &   1.006 &  0 & 0&& 0.005  &   $-$0.006  &    0.000 &&  6e$-$04  &   0.0012  &   8e$-$04\\
		&$-$0.5  & $-$0.502  &   1  &   1.003 &  2 & 2.002&& 0.002  &   $-$0.003  &   $-$0.002&&  5e$-$04  &   0.0014  &   0.0063 \\
		&$-$0.5  & $-$0.505  &   1  &   1     &  4 & 4&&  0.005  &    0.000  &   0.000 && 7e$-$04  &   0.0012  &   0.0226 \\
$1000$&$-$0.5  & $-$0.501  &   0  &   0.009 &  0 & 0.017&& 0.001  &   $-$0.009  &   $-$0.017&&0.001  &  0.0048  &    0.0114 \\
		&$-$0.5  & $-$0.505  &   0  &   0.005 &  2 & 2.01&& 0.005  &   $-$0.005  &   $-$0.010&&3e$-$04  &  0.0013  &    0.0065  \\
		&$-$0.5  & $-$0.505  &   0  &   0.004 &  4 & 4.005&&  0.005  &   $-$0.004  &   $-$0.005&&4e$-$04  &  0.0011  &    0.0223  \\
		&$-$0.5  & $-$0.505  &  $-$1  &  $-$0.985 &  0 & 0.023&& 0.005  &   $-$0.015  &   $-$0.023&& 4e$-$04  &  0.0049  &   0.0152\\
		&$-$0.5  & $-$0.503  &  $-$1  &  $-$1     &  2 & 2.02&& 0.003  &    0.000  &   $-$0.020 && 2e$-$04  &  0.0026  &   0.0204\\
		&$-$0.5  & $-$0.505  &  $-$1  &  $-$0.996 &  4 & 4.006 && 0.005  &   $-$0.004  &  $-$0.006&&1e$-$04  &  7e$-$04   &   0.0114\\
		\hline
	\end{tabular}\label{Tablexim05}
}
\end{table}

\end{document}